\documentclass[submission,copyright,creativecommons]{eptcs}
\providecommand{\event}{AiML 2026} 

\usepackage{aiml26}

\usepackage{iftex}

\usepackage{graphicx} \usepackage{amsfonts,amsmath,amssymb, amsthm, framed, mathtools, mathrsfs, float, stmaryrd, booktabs}
\usepackage[dvipsnames]{xcolor}
\usepackage{tikz-cd}
\usepackage{tikz}
\usetikzlibrary{arrows.meta}
\usepackage{cleveref}
\usepackage[numbers]{natbib}

\newcommand{\E}{\mathcal{E}}
\newcommand{\F}{\mathcal{F}}

\newcommand{\M}{\mathcal{M}}
\newcommand{\N}{\mathcal{N}}
\newcommand{\ra}{\rightarrow}
\newcommand{\Lang}{\mathcal{L}}

\ifpdf
\usepackage{underscore}         \usepackage[T1]{fontenc}        \else
\usepackage{breakurl}           \fi

\title{Belief Contraction in Dynamic Epistemic Logic}
\author{Gaia Belardinelli
	\institute{Department of Philosophy\\
		Stanford University\\
		Stanford, USA}
	\email{gaiabel@stanford.edu}
	\and
	Snow Zhang 
	\institute{Department of Philosophy\\
		University of California, Berkeley\\
		Berkeley, USA}
	\email{\quad snowzhang@berkeley.edu}
}

\newcommand{\titlerunning}{Belief Contraction in Dynamic Epistemic Logic}
\newcommand{\authorrunning}{G. Belardinelli \& S. Zhang}

\hypersetup{
	bookmarksnumbered,
	pdftitle    = {\titlerunning},
	pdfauthor   = {\authorrunning},
	pdfsubject  = {EPTCS},               	pdfkeywords = {keyword1, keyword2} }

\begin{document}
	\maketitle
\begin{abstract}
Dynamic epistemic logic (\textsf{DEL}) represents belief change via model transformations induced by epistemic events. Its standard formulation (Baltag, Moss, Solecki, 1998) provides a natural account of belief expansion through the elimination of possibilities, but it cannot model belief contraction about factual propositions. A classic response enriches Kripke models with plausibility orderings, representing contraction as an update that promotes certain possibilities over others. We show that this approach has expressive limitations. In particular, the approach cannot model belief that violates positive introspection and contraction dynamics in response to a \textit{hedged} public announcement that $\varphi$ \textit{might} be false. 
Motivated by these considerations, we introduce a 
mechanism for belief contraction defined directly on standard Kripke models, without 
any constraints on the doxastic accessibility relation. We show that it satisfies some of the standard properties of belief contraction but not others, study the conditions under which contraction may be unsuccessful, and provide a sound and complete axiomatization of the logic via reduction axioms. We also define a more general dynamic logic that is an extension of standard \textsf{DEL} and accommodates belief contractions due to events such as private or semi-private announcements, and provide a complete and sound axiomatization of the general logic.
\end{abstract}

	\section{Introduction}
Dynamic epistemic logic (\textsf{DEL}) is a family of logics that model multi-agent belief change by incorporating information revealed by epistemic events—such as public or private announcements—into agents’ epistemic states.
In standard \textsf{DEL} \cite{baltag1998logic}, this incorporation proceeds by \textit{eliminating} all possibilities that are incompatible with the disclosed information.
As a result, standard \textsf{DEL} can't straightforwardly model announcements that make agents reconsider possibilities they previously ruled out as impossible, as required for belief contraction. A popular solution is to distinguish ``hard'' from ``soft'' updates; while hard updates involve \textit{eliminating} possibilities, soft updates involve \textit{promoting} possibilities in terms of their comparative \textit{plausibility} \cite{benthem2007dynamic,baltag2008qualitative}. Then, belief contraction about a factual proposition $p$ consists in coming to judge that a $\neg p$-possibility is at least as plausible as the most plausible $p$-possibility, thereby losing the belief that $p$ holds.

While the plausibility-based approach can model belief contraction \cite{baltag2008qualitative,fiutek2013}, it has certain expressive limitations. First, since plausibility orderings are transitive, the framework validates positive introspection for beliefs (axiom \textbf{4}), and so can't model agents who are mistaken about their own beliefs. Second, the approach can't accommodate belief contractions due to \textit{hedged public announcements}---announcements that $p$ \textit{might be false}. Intuitively, such an announcement could cause an agent who initially believes $p$ to suspend judgment on $p$ without gaining any new (factual) beliefs. We show that there is a precise sense in which such dynamics \textit{can't} be modeled in the plausibility framework. 

These limitations motivate the search for a new model of belief contraction. We offer such a model in this paper. We introduce a dynamic operation of belief contraction on standard, multi-agent Kripke models of beliefs, with no assumptions on the doxastic accessibility relations. Roughly, a hedged announcement that \textit{$\varphi$ might be false} triggers a contraction on $\varphi$, which involves adding $\neg\varphi$-possibilities to an agent's belief set if she believes $\varphi$, and doing nothing otherwise. We define a sound and complete logic for this contraction modality, which we call \textit{hedged public announcement logic (\textsf{HPAL})}, and prove that this update satisfies a number of natural properties of belief contraction. 
We then define a more general dynamic logic \textsf{GDEL} that conservatively extends standard \textsf{DEL} to also accommodate belief contractions due to events such as private or semi-private announcements. We show that it contains \textsf{HPAL} and that its updates are always a simulation of a refinement of an initial Kripke model. Finally, we provide a complete and sound axiomatization of the general logic.

\noindent\textbf{Related literature.} The logic \textsf{GDEL} can be viewed as a natural extension of standard \textsf{DEL}. As contraction is a form of information-loss, our model is also related to logics of \textit{forgetting} and simulation modal logic \cite{ditmarsch2025simulation,ditmarsch2009introspective,fernandez-duque2015forgetting}, as well as to the fragment of graph modifier logic that adds edges to Kripke models \cite{aucher2009global}. Our theory however contrasts with the theory of belief contraction as studied in AGM \cite{alchurron1985logic}; our contraction operator does not satisfy many of the AGM postulates. This is to be expected, as dynamic contraction involves transformations that could change the truth-values of epistemic formulas, which is not considered in AGM. Philosophically, the problem of belief contraction due to hedged public announcement is related to the problem of evaluating indicative conditionals that have modal antecedents, which remains an open problem \citep{yalcin2007epistemic, holliday2017indicative, vanBenthem2023-VANTLO-47}. 

The proofs of main theorems can be found in the Appendix.

	\section{Standard \textsf{DEL} framework}\label{sec:standardDEL}
	We begin by reviewing the basics of \textsf{DEL} and fix notations.\footnote{For presentation purposes, in this section we only introduce the basic doxastic language without dynamic modalities. We introduce the full language of \textsf{DEL} in the last section, where we also use it. For an introduction to \textsf{DEL}, see \cite{sep-dynamic-epistemic}.} 
Let $At$ be a countable set of atomic formulas and $Ag$ be a finite set of agents. Let $\Lang$ be the \textit{doxastic language} generated from $At$ by the following grammar:
		$\varphi\coloneqq \top\mid p\mid\neg\varphi\mid\varphi\wedge\varphi\mid B_a\varphi$,
	where $p\in At$ and $a\in Ag$. We write $\bot$ for $\neg\top$ and $\widetilde{B}_a\varphi$ for $\neg B_a\neg\varphi$. A \textit{literal} is an atomic formula $p$ or its negation $\neg p$. Let $\bigwedge_{\varphi\in S}\varphi=\top$ if $S$ is empty.

We adopt the standard semantics of epistemic logic, where a Kripke model is a triple $\M=(W, R, V)$ with $W\not=\emptyset$ a set of worlds, $R_a\subseteq(W\times W)$ an accessibility relation defined for all $a\in Ag$, and $V: At\ra\wp(W)$ is a valuation function. We use notation $R_a(w)=\{v\in W\colon (w,v)\in R_a\}$ to represent \textit{the set of possible worlds compatible with agent $a$'s belief at $w$}.

The satisfaction relation between a Kripke model $\M=(W,R,V)$ with world $w\in W$ and a formula $\varphi\in\Lang$ is defined as standard, where in particular the semantics of the belief modality is given by:
    \vspace{-5pt}
		\begin{table}[H]
			    \centering
			    \begin{tabular}{lll}
				$\M, w\vDash B_a \varphi$ & iff & $\M, v\vDash\varphi$ for all $v\in R_a(w)$.
			\end{tabular}
			\end{table}
    \vspace{-10pt}
\noindent A formula $\varphi$ is \textit{valid in a model $\M$} (notation: $\M\vDash \varphi$), if for all worlds $w$ in $\M$, $\M,w\vDash \varphi$. A formula $\varphi$ is \textit{valid in a class of models $\mathfrak{L}$} (notation: $\vDash_\mathfrak{L} \varphi$), if for all models $\M$ in $\mathfrak{L}$, $\varphi$ is valid in $\M$. 

	    A \textit{standard event model} is a triple $\E=(E, Q, pre)$ where $E\not=\emptyset$ is a finite set of events, $Q_a\subseteq(E\times E)$ is an accessibility relation between events, for all $a\in Ag$, and $pre: E\ra\Lang$ a precondition function.\footnote{In the most general form of \textsf{DEL}, the precondition function assigns to each event a formula from the dynamic language (i.e. $\Lang$ extended with dynamic formulas). We keep it simple here and consider this general \textsf{DEL} form in the last section.} 
We use notation $Q_a(e)=\{f\in E\colon (e,f)\in Q_a\}$ for the set of events that are considered possible by $a$ at $e$. 

For a Kripke model $\M=(W, R, V)$ and a standard event model $\E=(E, Q, pre)$, the \textit{standard product update} of $\M$ with $\E$ is the Kripke model  $\mathcal{M} \otimes \mathcal{E} = (W^\E,R^\E,V^\E)$ where $W^\E=\{(w,e)\in W\times E\colon w\vDash pre(e)\}$, $R^\E_a(w,e)=\{(v,f)\in W^\E\colon v\in R_a (w)\text{ and }  f\in Q_a(e)\}$, 
and $V^\E(p)=\{(w,e)\in W^\E\colon w\in V(p)\}$, for all $p\in At$.
	\noindent Informally, a world $(w,e)$ represents \textit{the state of affairs at $w$ after $e$ occurs}. We call these event models and product updates \textit{standard }to distinguish them from the framework introduced later.

It is well known that standard \textsf{DEL} has undesirable consequences when agents are presented with information that contradicts their beliefs \cite{benthem2007dynamic,herzig2017dynamic}. For a concrete illustration:
\begin{example}\label{ex:alice1}
Alice believes that she locked her office door before she left. She runs into Bob, who tells her that \textit{her office door was open}. Alice trusts Bob and revises her beliefs accordingly.
\end{example}
\noindent Let $p$ denote the proposition ``Alice's office door is open''. Figure \ref{fig:DELlimitation} illustrates that, if we model Alice's belief revision using a standard product update, then she will end up having inconsistent beliefs.
\vspace{-5pt}
\begin{figure}[H]
\begin{minipage}{0.33\textwidth}
    \centering    \begin{tikzpicture}
         \node[label=left:$\M$] (A) {\begin{tikzpicture}[node distance=2cm, ->, >=Stealth, thick]
        \tikzstyle{state} = [circle, draw, thick, minimum size=0.65cm, inner sep=0pt, align=center]
            \node[state, label=below:$w$] (w) {\(\neg p\)};
  \node[state, right of=w, label=below:$v$] (v) {\(p\)};
        \path[->](v) edge node[below] {$a$} (w)
 (w) edge[loop left] node[left] {$a$} (w);
        \end{tikzpicture}};
    \end{tikzpicture}
\end{minipage}
\begin{minipage}{0.33\textwidth}
    \centering    \begin{tikzpicture}
         \node[label=left:$\E$] (A) {\begin{tikzpicture}[node distance=2cm, ->, >=Stealth, thick]
        \tikzstyle{event} = [rectangle, draw, minimum size=0.6cm, inner sep=0pt, align=center, thick]
            \node[event, label=below:$e$] (e) {\(p\)};
    \path[->] (e) edge[loop right] node[right] {$a$} (e);
        \end{tikzpicture}};
    \end{tikzpicture}
\end{minipage}
\hspace{-1cm}
\begin{minipage}{0.33\textwidth}
    \centering
\begin{tikzpicture}
         \node[label=left:$\M\otimes\E$] (N) {\begin{tikzpicture}[node distance=2cm, ->, >=Stealth, thick]
        \tikzstyle{state} = [circle, draw, minimum size=0.6cm, inner sep=0pt, align=center, thick]
\node[state, label=below:${(v,e)}$] (v) {\(p\)};
        \end{tikzpicture}};
\end{tikzpicture}
\end{minipage}
\vspace{-10pt}
        \caption{\textbf{Left}:         The Kripke model $\M=(W,R,V)$ representing 
        Alice's
        initial belief state. An edge from a world $v$ to a world $w$ means $w\in R_a(v)$. We have $\M\vDash B_a\neg p\wedge \neg B_a p$. 
        \textbf{Middle}: The standard event model $\E=(E,Q,pre)$ representing Bob's announcement. The event $e$ contains its own precondition $pre(e)=p$. The loop on $e$ means $e\in Q_a(e)$.
        \textbf{Right}: 
        The Kripke model $\M\otimes\E=(W^\E,R^\E,V^\E)$ representing 
        Alice's belief state after the update, 
        where $R^\E_a((v,e))=\emptyset$ and so e.g. $\M\otimes\E        \vDash B_ap\wedge B_a\neg p$.}
    \label{fig:DELlimitation}
\end{figure}
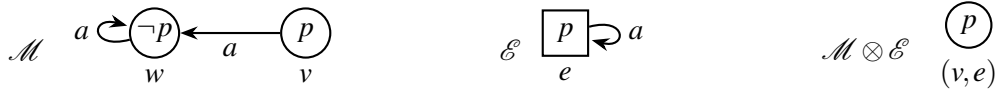

\section{Plausibility models}
The ``standard diagnosis'' of the problem of modeling belief contraction in standard \textsf{DEL} is that we need a richer view of beliefs \cite{benthem2007dynamic}. In particular, we need a framework where an agent believes not just any proposition entailed by her information,
but those that are true at the `best’ or `most plausible' possibilities compatible with her information \cite{baltag2008qualitative,benthem2007dynamic}. I believe that I am not dreaming. It is not \textit{impossible} that I am, but those possibilities are less \textit{plausible} than those in which I am awake. Similarly, one interpretation of Alice's initial doxastic situation is that she believes that her office is closed, not because she has completely ruled out the possibility that it is open, 
but only that she judges that possibility to be relatively implausible. And it is this relative plausibility judgment that gets revised by Bob's testimony.

This idea is formalized by modeling belief in terms of plausibility orderings over possible worlds.
Formally, a \emph{plausibility order on a set $S$} is a connected,  reflexive and transitive relation $\trianglelefteq\subseteq S\times S$ such that every non-empty subset $S'\subseteq S$ has at least one $\trianglelefteq$-minimal element, i.e. $Min_\trianglelefteq (S')=\{w\in S': \forall v\in S', w\trianglelefteq v\}\neq\emptyset$.   \footnote{Note that we assume the plausibility order to be \textit{connected}. This is more restrictive than usual \cite{baltag2008qualitative}, but it is justified here as our framework only models belief and omits knowledge.}
A \emph{plausibility model} is a Kripke model $(W,\preceq, V)$ with $\preceq_a$ a plausibility order on $W$, for all agents $a\in Ag$. We write $\prec_a$ for its asymmetric component, i.e.~$w\prec_a v$ if $w\preceq_a v$ and $v\not\preceq_a w$. The semantics of $\Lang$ over plausibility models is standard for propositional formulas, but belief is interpreted differently from above, capturing that the agent believes what is the case at the most plausible worlds:
$$\M, w\vDash B_a \varphi  \text{ iff } \M, v\vDash\varphi \text{ for all } v\in Min_{\preceq_a}(W).$$
A  \textit{plausibility event model} is an event model $\E=(E, \leq, pre)$ with $\leq_a$ a plausibility order on $E$ for all $a\in Ag$. We write $<_a$ for its asymmetric component. As in standard \textsf{DEL}, event models update plausibility models via product update. The \textit{plausibility-update} of $\M=(W, \preceq, V)$ with $\E$ is the plausibility model $\M\ast\E=(W^\E,\preceq^\E,V^\E)$ where $W^\E$ and $V^\E$ are defined as in standard product updates and $\preceq^\E$ is given by $(w,e)\preceq_a^\E(v,f) \text{ iff } e<_af, \text{ or } e\equiv_a f  \text{ and }w\preceq_a v$.
Figure \ref{fig:PlausibilityLimitation} shows how to model Example \ref{ex:alice1} via a plausibility update. Alice's change of mind is captured as a \textit{soft} update where no world is ruled out as impossible, but their relative plausibility is switched.\vspace{-4pt}
\begin{figure}[H]
\begin{minipage}{0.33\textwidth}
    \centering    \begin{tikzpicture}
         \node[label={[label distance=-2mm]left:$\M$}] (A) {\begin{tikzpicture}[node distance=2cm, ->, >=Stealth, thick]
        \tikzstyle{state} = [circle, draw, thick, minimum size=0.6cm, inner sep=0pt, align=center]
            \node[state, label=below:$w$] (w) {\(\neg p\)};
  \node[state, right of=w, label=below:$v$] (v) {\(p\)};
   \path[->](v) edge node[below] {$a$} (w)
(w) edge[loop left] node[below,xshift=-3pt] {$a$} (w)
(v) edge[loop right] node[below,xshift=3pt] {$a$} (v);
        \end{tikzpicture}};
    \end{tikzpicture}
\end{minipage}
\begin{minipage}{0.33\textwidth}
    \centering    \begin{tikzpicture}
         \node[label={[label distance=-2mm]left:$\E$}] (A) {\begin{tikzpicture}[node distance=2cm, ->, >=Stealth, thick]
        \tikzstyle{event} = [rectangle, draw, minimum size=0.6cm, inner sep=0pt, align=center, thick]
            \node[event, label=below:$e$] (e) {\(p\)};
            \node[event, label=below:$f$, left of=e] (f) {\(\neg p\)};
\path[->] (f) edge node[below] {$a$} (e)
(e) edge[loop right] node[below,xshift=3pt] {$a$} (e)
(f) edge[loop left] node[below,xshift=-3pt] {$a$} (f);
        \end{tikzpicture}};
    \end{tikzpicture}
\end{minipage}
\begin{minipage}{0.33\textwidth}
    \centering
\begin{tikzpicture}
         \node[label={[label distance=-2mm]left:$\M\ast\E$}] (C) {\begin{tikzpicture}[node distance=2cm, ->, >=Stealth, thick]
        \tikzstyle{state} = [circle, draw, thick, minimum size=0.6cm, inner sep=0pt, align=center]
            \node[state, label=below:${(w,f)}$] (w) {\(\neg p\)};
  \node[state, right of=w, label=below:${(v,e)}$] (v) {\( p\)};
   \path[->]
 (w) edge node[below] {$a$} (v)
(w) edge[loop left] node[below,xshift=-3pt] {$a$} (w)
(v) edge[loop right] node[below,xshift=3pt] {$a$} (v);
        \end{tikzpicture}};
    \end{tikzpicture}
\end{minipage}
\vspace{-10pt}
        \caption{\textbf{Left}: The plausibility model $\M=(W,\preceq,V)$ representing Alice's initial belief state. An edge from $v$ to $w$ means         $w\preceq_a v$. 
        We have $\M\vDash B_a \neg p\wedge \neg B_a p$. \textbf{Middle}: The plausibility event model $\E=(E,\le,pre)$ capturing a soft update that $p$ holds.  
                An edge from $f$ to $e$ means $e\le_af$. 
        \textbf{Right}: The plausibility model $\M\ast\E$ representing Alice's belief state after updating,
        with $\M\ast\E\vDash B_a p\wedge \neg B_a\neg p$.}
    \label{fig:PlausibilityLimitation}
\end{figure}
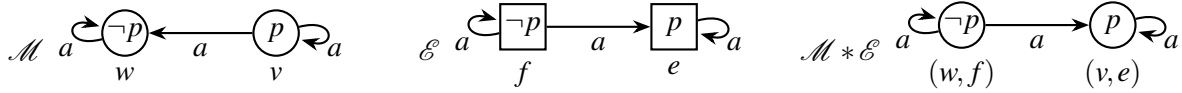
\subsection{Limitations of plausibility models}
One limitation of the plausibility model of beliefs is that it validates axiom \textbf{4}: $B_a\varphi\rightarrow B_aB_a\varphi$. For this reason, it cannot model agents who are mistaken about their own beliefs. Such cases, however, seem possible: an employer might profess that they believe in gender equality while consistently favoring one gender group in hiring.
One way of describing this employer is that they in fact believe that candidates from one gender group are more qualified, even though they (falsely) believe that they don't believe that.

Another limitation of the approach, mentioned above, concerns belief contraction that does not involve adopting any new factual beliefs. 
A variant of Example~\ref{ex:alice1} illustrates what we have in mind:
\begin{example}\label{eg:mightopen} As before, Alice believes that she locked the door to her office before she left. She runs into Bob, who tells her that \textit{her office door might be open}. Given Bob's testimony, Alice comes to suspend judgment on whether she locked the door to her office before she left.
\end{example}
\noindent Let $p$ denote the proposition \textit{Alice's office door is open}. Alice's belief states before and after the update 
can be modeled by $\M$ and $\N$ in Figure~\ref{fig:limitations}, respectively. \vspace{-4pt}
\begin{figure}[H]
\begin{minipage}{0.48\textwidth}
    \centering    \begin{tikzpicture}
         \node[label=left:$\M$] (A) {\begin{tikzpicture}[node distance=2cm, ->, >=Stealth, thick]
        \tikzstyle{state} = [circle, draw, minimum size=0.6cm, inner sep=0pt, align=center]
            \node[state, label=below:$w$] (w) {\(\neg p\)};
  \node[state, right of=w, label=below:$v$] (v) {\(p\)};
   \path[->] (v) edge node[below] {$a$} (w)
   (w) edge[loop left] node[left] {$a$} (w)
   (v) edge[loop right] node[right] {$a$} (v);
        \end{tikzpicture}};
    \end{tikzpicture}
\end{minipage}
\begin{minipage}{0.48\textwidth}
    \centering
\begin{tikzpicture}
         \node[label=left:$\N$] (N) {\begin{tikzpicture}[node distance=2cm, ->, >=Stealth, thick]
        \tikzstyle{state} = [circle, draw, minimum size=0.6cm, inner sep=0pt, align=center]
\node[state, label=below:$w$] (w) {\(\neg p\)};
  \node[state, right of=w, label=below:$v$] (v) {\(p\)};
  \path[<->] (v) edge node[below] {$a$} (w);
  \path[->]
(w) edge[loop left] node[left] {$a$} (w)
(v) edge[loop right] node[right] {$a$} (v);
        \end{tikzpicture}};
\end{tikzpicture}
\end{minipage}
\vspace{-8pt}
        \caption{\textbf{Left}: The         plausibility model $\M=(W,\preceq,V)$ representing Alice's initial belief state, where $\M\vDash B_a \neg p \wedge \neg B_a p$. Conventions for edges are as in Fig.~\ref{fig:PlausibilityLimitation}. \textbf{Right}: The 
                plausibility model $\N=(W,\preceq^\N,V)$ representing Alice's belief state after updating on Bob's testimony, where $\N\vDash \neg B_a p \wedge \neg B_a\neg p$.}
    \label{fig:limitations}
\end{figure}
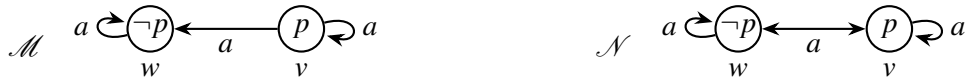
\noindent 
However, it turns out that this update \textit{cannot} be modeled as a plausibility update. 
More precisely, let $\mathcal{K} \equiv \mathcal{K}'$ denote modal equivalence between two models $\mathcal{K}$ and $\mathcal{K}'$ with respect to $\Lang$.\footnote{Recall that two Kripke (or plausibility) models $\mathcal{K}=(W,R,V)$ and $\mathcal{K}'=(W',R',V')$ are \emph{modally equivalent} with respect to $\Lang$ if for every formula $\varphi \in \Lang$ and every state $w \in W$, there exists a state $w' \in W'$ such that $\mathcal{K},w \vDash \varphi$ iff $\mathcal{K}',w' \vDash \varphi$, and conversely, for every $w' \in W'$ there exists a state $w \in W$ with the same property.} Then:
\begin{theorem}\label{prop:nome}
	   There is no plausibility event model $\E$ such that $\M\ast\E\equiv\N$.
\end{theorem}
	\begin{proof}
	Let $\M=(W, \preceq, V)$ and $\N=(W, \preceq^\N, V)$ be as depicted in Figure \ref{fig:limitations}. Suppose towards a contradiction that there is a plausibility event model $\E=\{E, \leq, pre\}$ such that $\M\ast\E=(W^\E, \preceq^\E, V^\E)$ is a plausibility model that is modally equivalent to $\N$. We claim that, for any $(v,f)\in W^\E$, there exists a $(v,f')\neq (v,f)$ such that $(v,f')\prec_a^\E(v,f)$. It follows then that $\M\ast\E$ contains an infinite descending chain, which contradicts the assumption that it is a plausibility model. 
Let $(v,f)\in W^\E$. By assumption, $(\M\ast\E, (v,f))$ is modally equivalent to $(\N, v)$. Since $(\N, v)\vDash \neg B_a p$ , we have $(\M\ast\E, (v,f))\vDash\neg B_a p$. So there exists $(w,e)\in W^\E$ such that $(w,e)\preceq_a^\E(v,f)$. For this $(w,e)$, it must be that $(\M\ast\E, (w,e))$ is modally equivalent to $(\N, w)$ (since it can't be modally equivalent to $(\N, v)$). Since $\N, w\vDash\widetilde{B}_ap$, there exists $(v,f')\in W^\E$ such that $(v,f')\preceq_a^\E(w,e)$. Moreover, since $v\not\preceq_a w$, we must have $f'<_ae\leq_a f$. So $f'<_af$ and therefore $(v,f')\prec_a^\E(v,f)$.		\end{proof}
    \noindent Note that if 
    two models are not modally equivalent with respect to $\Lang$, then they are not modally equivalent with respect to any extension of $\Lang$. So it follows from \Cref{prop:nome} that there is no plausibility event model $\E$ such that $\M\ast\E$ is modally equivalent to $\N$ with respect to $\Lang$ enriched with, e.g. modalities for conditional beliefs, safe beliefs or dynamic modalities. Relatedly, since modal equivalence is generally assumed to be a necessary condition for an adequate notion of bisimilarity, it also follows that there is no plausibility event model $\E$ such that $\M\ast \E$ is bisimilar to $\N$ with respect to any adequate notion of bisimilarity (see e.g. \cite{demey2011remarks,andersen2013bisimulation,andersen2017bisimulation} for different notions of bisimilarity for plausibility structures).

\section{A new model for belief contraction}\label{sec:4contraction}
\Cref{prop:nome} shows that plausibility updates can't weaken an agent's belief in a way that makes her consider the two worlds equiplausible. We now introduce a framework for belief contraction that can achieve that. In general, belief contraction can be caused by many different kinds of epistemic events. In this section, we focus on a simple case of belief contraction due to \textit{hedged public announcements}, like Bob's testimony that Alice's door \textit{might} be open.\footnote{The name \textit{hedged public announcement} reflects the public and non-committal (``hedged'') character of the announcement, which opens up possibilities rather than restricting them---much like epistemic modals in so-called ``hedged assertions'' \cite{benton2020hedged}.} 
Formally, we represent the effects of such announcements by the dynamic modality ``$[\div\varphi]\psi$''. Intuitively, it says that $\psi$ is the case after the hedged public announcement that $\varphi$ \textit{might be false}. Let the \textit{contraction language} $\Lang^\div$ be the language generated by the grammar of $\Lang$ extended with the following clauses for the dynamic modality and a universal modality: $\varphi\coloneqq [\div\varphi]\varphi\mid \forall \varphi$.
We define $\exists \varphi\coloneqq \neg\forall\neg\varphi$.

A standard (non-hedged) public announcement that \textit{$\varphi$ is false} is modeled as eliminating all the $\varphi$-possibilities fro5m all agents' doxastic spaces \cite{plaza2007logics}. A natural thought, then, is that a hedged public announcement that \textit{$\varphi$ might be false} should add all the $\neg\varphi$-possibilities to all agents’ doxastic spaces.
However, this is not quite right. Adding all such possibilities to all agents' doxastic spaces may force those who already consider $\neg\varphi$ possible to update their beliefs and start considering additional $\neg\varphi$-possibilities, thereby losing more beliefs than is warranted by the announcement.
This suggests that, unlike standard public announcements, the epistemic 
effects of a hedged announcement that $\varphi$ might be false depend on the agents' initial beliefs, and in particular on whether they believe $\varphi$ prior to the announcement. For this section, we assume that this is the only relevant difference: if an agent already considers $\neg\varphi$ possible, then the announcement \textit{$\varphi$ might be false} has no effects on her beliefs; if the agent believes $\varphi$, then the announcement leads her to
consider all $\neg\varphi$-worlds as possible. 

	\begin{definition}[Contraction on $\varphi$] \label{def:contraction}Given a Kripke model $\M=(W,R,V)$ and a formula $\varphi\in\Lang^\div$. Let $\M^{\div\varphi}=(W, R^{\div\varphi}, V)$ be defined by, for all $a\in Ag$ and all $w\in W$:

			$R_a^{\div\varphi}(w)=\begin{cases}
				R_a(w)\cup \{v\in W\colon \M,v\vDash \neg \varphi\} & \text{ if } \M,w\vDash B_a\varphi\\
				R_a(w) & \text{ otherwise}
			\end{cases}$ 
		
	\end{definition}
	\noindent Call $\M^{\div\varphi}$ \textit{the model obtained from $\M$ by contracting on $\varphi$}. Note that it is always the case that $R_a(w)\subseteq R_a^{\div\varphi}(w)$.     This means that a hedged public announcement always \textit{extends} the initial Kripke model. 
    	The semantics of $\Lang^\div$ is given by the semantics of $\Lang$ over Kripke models extended with these clauses:	\begin{table}[H]
	   \centering
            \begin{tabular}{lll}
				$\M, w\vDash [\div\varphi]\psi$ & iff &$\M^{\div\varphi}, w\vDash\psi$;\\
                $\M, w\vDash \forall\varphi$ & iff & for all $v\in W$, $\M, v\vDash\varphi$.
            \end{tabular}
	\end{table}
    \vspace{-10pt}
\noindent Consider again Example~\ref{eg:mightopen}. In the present framework, the effect of Bob's announcement that \textit{Alice's office door might be open} can be modeled by contracting on the proposition \textit{Alice's office door is closed} ($\neg p$), which amounts to adding world $v$ to Alice's set of doxastic possibilities at $w$ and $v$. See Figure~\ref{fig:contraction1}.
 
 \vspace{-4pt}
 \begin{figure}[h]
\begin{minipage}{0.48\textwidth}
    \centering    \begin{tikzpicture}
         \node[label=left:$\M$] (A) {\begin{tikzpicture}[node distance=2cm, ->, >=Stealth, thick]
        \tikzstyle{state} = [circle, draw, minimum size=0.6cm, inner sep=0pt, align=center]
            \node[state, label=below:$w$] (w) {\(\neg p\)};
  \node[state, right of=w, label=below:$v$] (v) {\(p\)};
 \path[->](v) edge node[below] {$a$} (w)
(w) edge[loop left]  node[left] {$a$} (w);
        \end{tikzpicture}};
    \end{tikzpicture}
\end{minipage}
\begin{minipage}{0.48\textwidth}
    \centering
    
    \begin{tikzpicture}
	         \node[label=left:$\M^{\div \neg p}$] (A) {\begin{tikzpicture}[node distance=2cm, ->, >=Stealth, thick]
	        \tikzstyle{state} = [circle, draw, minimum size=0.6cm, inner sep=0pt, align=center]
	            \node[state, label=below:$w$] (w) {\(\neg p\)};
	  \node[state, right of=w, label=below:$v$] (v) {\(p\)};
 \path[->](v) edge[bend left] node[below] {$a$} (w)
 (w) edge[bend left, red] node[below] {$\textcolor{red}{a}$} (v)
(w) edge[loop left] node[left] {$a$} (w)
(v) edge[loop right, red] node[right] {$\textcolor{red}{a}$} (v);
	        \end{tikzpicture}};
	    \end{tikzpicture}
\end{minipage}
\vspace{-10pt}
\caption{\textbf{Left}: The Kripke model $\M$ (same as Figure~\ref{fig:DELlimitation}), representing Alice's initial belief state.
\textbf{Right}: The Kripke model $\M^{\div \neg p}$ representing Alice's belief state after updating on Bob's hedged announcement, obtained by contracting on $\neg p$. We have $\M^{\div \neg p}\vDash \neg B_a\neg p\wedge \neg B_a p$, so Alice lost a belief in $\neg p$.}\label{fig:contraction1}
	\end{figure}
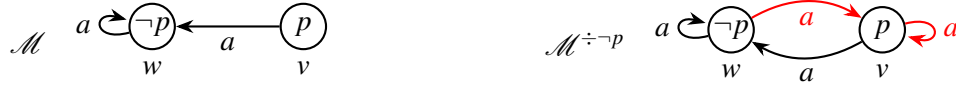

One natural question is what is the logic of the contraction update as defined above. The next theorem shows that such logic is given by the system presented in Table~\ref{tab:HPAL}. 
\begin{table}[H]
			\centering
            \small
			\begin{tabular}{cc}
            \hline
												                CL & all classical propositional tautologies\\
				$\forall$K & $\forall(\varphi\ra\psi)\ra (\forall\varphi\ra\forall\psi)$\\
                $\forall$T & $\forall\varphi\ra \varphi$\\
                $\forall$5 & $\neg\forall\varphi\ra \forall\neg\forall\varphi$\\
                $\forall$B & $\forall\varphi\ra B_a\varphi$\\
				BK & $B_a(\varphi\ra\psi)\ra (B_a\varphi\ra B_a\psi)$\\
                  \midrule
				A1 & $[\div\varphi]p\leftrightarrow p$\\
				A2 & $[\div\varphi]\neg\psi\leftrightarrow \neg [\div\varphi]\psi$\\
				A3 & $[\div\varphi](\psi\wedge\chi)\leftrightarrow[\div\varphi]\psi\wedge [\div\varphi]\chi$\\
				A4& $[\div\varphi]\forall\psi\leftrightarrow \forall[\div\varphi]\psi$\\
				A5 & $[\div\varphi]B_a\psi\leftrightarrow (B_a[\div\varphi]\psi\wedge (B_a\varphi\rightarrow \forall (\neg\varphi\rightarrow[\div\varphi]\psi)))$\\
                \midrule
                MP & From $\varphi$ and $\varphi\rightarrow \psi$, infer $\psi$\\
                RE & From $\varphi\leftrightarrow\psi$, infer $\chi[\varphi/p]\leftrightarrow\chi[\psi/p]$\\
               NEC & {Necessitation rules for all modalities}\footnotemark\\
				\hline
			\end{tabular}
			\caption{\small The Theory of \textsf{Hedged Public Announcement Logic (HPAL)}}
			\label{tab:HPAL}
		\end{table}
                      
                     \footnotetext{Notice that necessitation for the belief modality is redundant, as it is derivable from $\forall$B and $\forall$-necessitation.}
                            
		\begin{theorem}\label{thm:HPALcompleteness}
			The dynamic logic of contraction updates is completely axiomatized by \textsf{HPAL}.
		\end{theorem} 
  
           \begin{proof}[Proof sketch.] Soundness follows by straightforward validity arguments.
For completeness, given the reduction axioms, every formula with dynamic modalities is provably equivalent in      \textsf{HPAL} to a formula without.
          So the completeness of      \textsf{HPAL} follows from the completeness of the basic doxastic logic with the universal modality \citep[Theorem 7.2]{bluebook}.    \end{proof} 
     \noindent Besides axioms and rules of inference of the basic modal logic \textsf{K} with the universal modality (which is an S5 modality),      the theory contains inference rules and reduction axioms for the contraction modality that parallel the reduction axioms in Public Announcement Logic (\textsf{PAL}) \cite{plaza2007logics,van2008dynamic,wang2013axiomatizations}.      Axiom A5 is specific to our logic of contraction. It says that after contraction on $\varphi$ (viz.~a hedged announcement that $\varphi$ might be false), an agent believes $\psi$ iff she doesn't believe $\varphi$ prior to the announcement and she believes that $\psi$ holds after the announcement, or she believes $\varphi$ prior to the announcement and at any $\neg\varphi$-worlds, $\psi$ holds after the announcement. 
 \section{Preservation and Moorean formulas}
A well-known phenomenon in public announcement logic is that some formulas become false after they are announced. As a result, the  formula is not believed after the announcement. In \textsf{PAL}, formulas that become false or are not believed after an announcement are both called \textit{unsuccessful formulas}. Let $[!\varphi]$ be the dynamic modality ``after the announcement of $\varphi$'' in \textsf{PAL}.\footnote{In \textsf{PAL}, the formula $[!\varphi]\psi$ intuitively says $\psi$ is true after eliminating all $\neg\varphi$-possibilities from the model; see \cite[Chapter 4]{van2008dynamic} and references therein for more on \textsf{PAL}.} There are thus two closely related notions of \textit{success} in \textsf{PAL}: (i) what is announced is true after its announcement ($[!\varphi]\varphi$ is valid); (ii) what is announced is believed after its announcement ($[!\varphi]B_a\varphi$ is valid for all $a\in Ag$). For clarity, we will say $\varphi$ is \textit{preserved in \textsf{PAL}} if it satisfies (i); $\varphi$ is \textit{successful in \textsf{PAL}} if it satisfies (ii).
It is well-known that, in \textsf{PAL}, any formula that satisfies (i) also satisfies (ii), though the converse does not hold.\footnote{The proof that (i) implies (ii) is analogous to the ``only if''-direction of \citep[Proposition 4.32]{van2008dynamic}, though here the modality is belief rather than knowledge, and in particular it does not necessarily satisfy the T-axiom, which is the reason why the other direction does not hold in general. For example, suppose $Ag=\{a\}$, let $\varphi\coloneqq p\wedge B_a\neg p\wedge \neg B_a\bot$, and consider a Kripke model and a world $w$ in it where such $\varphi$ is true. Then, after the public announcement of $\varphi$, which eliminates all possibilities where $\varphi$ is not true and so also all worlds that were accessible for $a$ at $w$, $a$ will have inconsistent beliefs at $w$ and believe (trivially) that $\varphi$ is true. Hence, $[!\varphi]B_a\varphi$ is valid. On the other hand, $[!\varphi]\varphi$ is invalid and in particular $[!\varphi]\neg B_a\bot$ is invalid.} In this section, we focus on developing a notion of preservation for \textsf{HPAL} that is analogous to (i) for \textsf{PAL}.

A natural starting point is to say that $\varphi$ is \textit{preserved} if it is true after the hedged announcement that $\varphi$ might be true, viz.~$[\div\neg\varphi]\varphi$ is valid. But this isn't quite right: while a public announcement that $\varphi$ \textit{is} true \textit{eliminates} all $\neg\varphi$-possibilities, a hedged public announcement that $\varphi$ \textit{might} be true doesn't eliminate any possibilities, and so any possibility where $\varphi$ is false before the announcement remains in the model (though they may satisfy $\varphi$ in the new model after the announcement). This observation suggests the following alternative definition of preservation for hedged public announcements
            (let $\mathfrak{L}$ be a class of models that validates logic \textbf{L}):
\begin{definition}[Preservation]\label{def:preservation}
   $\varphi$ is \emph{preserved} (in the logic \textbf{L}) if $\vDash_\mathfrak{L}\varphi\rightarrow [\div\neg\varphi]\varphi$
\end{definition}
 
\noindent 
Not all formulas in $\Lang^\div$ are preserved in \textsf{HPAL}. Consider Example \ref{eg:mightopen} again. Suppose Bob tells Alice instead that \textit{she might have a false belief that her office door is closed}, viz.~it might be true that $p\wedge B_a\neg p$. Figure \ref{fig:moore} shows that this hedged announcement is not preserved. 
\vspace{-4pt}
 \begin{figure}[H]
\begin{minipage}{0.48\textwidth}
    \centering    \begin{tikzpicture}
         \node[label=left:$\M$] (A) {\begin{tikzpicture}[node distance=2cm, ->, >=Stealth, thick]
        \tikzstyle{state} = [circle, draw, minimum size=0.65cm, inner sep=0pt, align=center]
            \node[state, label=below:$w$] (w) {\(\neg p\)};
  \node[state, right of=w, label=below:$v$] (v) {\(p\)};
 \path[->](v) edge node[below] {$a$} (w)
(w) edge[loop left] node[left] {$a$} (w);
        \end{tikzpicture}};
    \end{tikzpicture}
\end{minipage}
\begin{minipage}{0.48\textwidth}
    \centering
    \begin{tikzpicture}
	         \node[label=left:$\M^{\div \neg (p\wedge B_a\neg p)}$] (A) {\begin{tikzpicture}[node distance=2cm, ->, >=Stealth, thick]
	        \tikzstyle{state} = [circle, draw, minimum size=0.65cm, inner sep=0pt, align=center]
	            \node[state, label=below:$w$] (w) {\(\neg p\)};
	  \node[state, right of=w, label=below:$v$] (v) {\( p\)};
 \path[->](v) edge[bend left] node[below] {$a$} (w)
 (w) edge[bend left, red] node[below] {$\textcolor{red}{a}$} (v)
(w) edge[loop left] node[left] {$a$} (w)
(v) edge[loop right, red] node[right] {$\textcolor{red}{a}$} (v);
	        \end{tikzpicture}};
	    \end{tikzpicture}
\end{minipage}
\vspace{-10pt}
\caption{Alice's belief update as a contraction on $\neg (p\wedge B_a\neg p)$. The formula $p\wedge B_a\neg p$ is true at $v$ before the update (on the left) but false at $v$ after the update (on the right). Hence, $p\wedge B_a\neg p$ is not preserved.}\label{fig:moore}
	\end{figure}
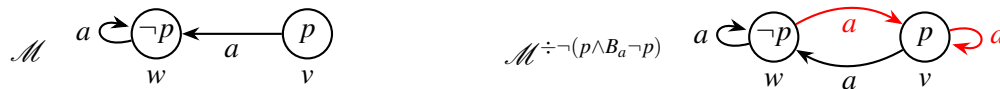
\noindent More generally, say $\varphi$ is \textit{self-refuting} (in the logic \textbf{L}) if $\vDash_\mathfrak{L} \varphi\rightarrow [\div\neg\varphi]\neg\varphi$. Clearly, $\varphi$ is not preserved if it is self-refuting. The above example is a special case of the following result:
\begin{proposition}\label{prop:existmoorean}
    The formulas $p\wedge \exists B_a\neg p$ and  $p\wedge B_a\neg p$ are self-refuting in \textsf{HPAL}.\footnote{Note that the     Moorean sentence $p\wedge B_a\neg p$ is unsatisfiable in \textbf{S5}, while     $p\wedge\exists B_a\neg p$ is satisfiable even if belief is factive.}
\end{proposition}
\begin{proof} Let $\varphi\coloneqq p\wedge\exists B_a\neg p$. Suppose $\M, w\vDash p\wedge\exists B_a\neg p$. Then $\M\vDash\exists B_a\neg p$ and so $\M^{\div\neg\varphi}=\M^{\div\neg p}$. Thus $\M^{\div\neg\varphi}\vDash \forall\widetilde{B}_ap$ and so $\M^{\div\neg\varphi}, w\vDash \neg\varphi$. Similarly, let $\psi\coloneqq p\wedge B_a\neg p$. Suppose $\M, w\vDash \psi$. Since $B_a\neg p\vDash B_a\neg \psi$, $\M, w\vDash B_a\neg\psi$. So $w\in R^{\div\neg\psi}_a(w)$. Since $\M, w\vDash p$, $\M^{\div\neg\psi}, w\vDash \widetilde{B}_a p$. Given that $\widetilde{B}_a p\vDash\neg\psi$, $\M^{\div\neg\psi}, w\vDash \neg\psi$.
\end{proof}
\noindent So which formulas are preserved? Recall that a formula $\varphi$ is \textit{existential} if it is built only from literals, $\wedge$, $\vee$ and $\widetilde{B}_a$; $\varphi$ is \textit{universal} if it is built only from literals, $\wedge$, $\vee$ and $B_a$. It's well-known that universal formulas remain true after being (standardly) publicly announced,
in any normal modal logic \cite{van2008dynamic}. Their dual---existential formulas---are always preserved by hedged public announcements:
\begin{theorem} \label{thm:existential}
If $\varphi\in\Lang^\div$ is equivalent in \textsf{HPAL} to an existential formula, then it is preserved in \textsf{HPAL}.
\end{theorem}
\begin{proof}
    Suppose $\varphi$ is equivalent in \textsf{HPAL} to an existential formula. Then $\varphi$ is preserved under model extension \cite{andreka1998modal}. Since $\M^{\div\neg\varphi}$ extends $\M$, if $\M, w\vDash\varphi$, then $\M^{\div\neg\varphi}, w\vDash\varphi$, i.e.~$\M, w\vDash \varphi\ra[\div\neg\varphi]\varphi$. 
\end{proof}
\noindent One upshot of Theorem \ref{thm:existential} is that, while the ``strong'' Moorean sentence $p\wedge B_a\neg p$ is never preserved, the ``weak'' Moorean sentence $p\wedge\neg B_ap$, which is equivalent to an existential formula, is preserved. 

So being logically equivalent to an existential formula is sufficient for being preserved. But it is not necessary. In particular, the non-existential formula $p\wedge B_ap$ (agent $a$ has a true belief that $p$) is preserved.
\begin{proposition}\label{prop:pres}
      The formula $p\wedge B_ap$ is preserved in \textsf{HPAL}. \end{proposition}
\begin{proof} Let $\M=(W, R, V)$ be a Kripke model with $\M,w\vDash p\wedge B_ap$. As $R_a^{\div\neg(p\wedge B_ap)}(w)\subseteq R_a(w)\cup\{v\in W: \M, v\vDash p\wedge B_ap\}\subseteq \{v\in W: \M, v\vDash p\}$ and $p$ is preserved, then $\M^{\div\neg(p\wedge B_ap)}, w\vDash p\wedge B_ap$.
\end{proof}
\noindent It would be nice to have a complete characterization of the set of preserved formulas, which we leave as an open question. 
       But there is a partial result. Let's consider \textit{knowledge} instead of \textit{belief} for a moment, and restrict our attention to single-agent models. The standard logic for knowledge is \textbf{S5}, and extending \textsf{HPAL} with \textbf{S5} axioms implies that all formulas in $\Lang$ are preserved.
\begin{theorem}
   Suppose $|Ag|=1$. Then every formula in $\Lang$ is preserved in \textsf{HPAL} extended with \textbf{S5}.
\end{theorem}
\begin{proof}[Proof sketch.]
In \textbf{S5}, if $\M,w\vDash \varphi$ with $\varphi\in\Lang$, then every world in the equivalence class of $w$ satisfies $\widetilde{B}_a\varphi$ and so contracting on $\neg\varphi$ does not change relations inside those classes. Then, $(\M,w)$ and $(\M^{\div\neg\varphi},w)$ are bisimilar, and since $\Lang$ is bisimulation-invariant, $\varphi$ is preserved.
\end{proof}
\noindent Note that model $\M$ in Figure~\ref{fig:moore} validates \textbf{KD45}, so the preservation result  doesn't hold in these settings.\footnote{As an anonymous referee pointed out, the result does not generalize to the multi-agent setting.}

\section{AGM-style contraction principles}\label{sec:AGM}
In the AGM literature, belief contraction is modeled as a static operation on the set of formulas, representing the agent's belief base, and is characterized by a set of axioms. As \cite{van2007dynamic} points out, many of those axioms do not naturally hold in the dynamic setting where belief contraction induces model transformations. In this section, we give an overview of which of the AGM axioms hold/fail for our contraction operator. Following \cite{segerberg1999two}, we render AGM statements that express that the agent believes $\chi$ after contracting by $\varphi$, by the formula $[\div\varphi]B_a\chi$. The modal version of the AGM axioms for contraction can then be expressed by the following formulas, where $\varphi,\psi,\chi\in\Lang^{\div}$:

\begin{itemize}
\item \textbf{Closure}. $[\div\varphi](B_a(\psi\ra\chi)\ra (B_a\psi\ra B_a\chi))$. 

\item\textbf{Success}. $\exists\neg \varphi\rightarrow [\div\varphi]\neg B_a\varphi$. 

\item \textbf{Inclusion}. $[\div\varphi]B_a\psi \rightarrow B_a\psi$. 

\item \textbf{Vacuity}. $\neg B_a\varphi \rightarrow (\chi\leftrightarrow [\div\varphi]\chi)$. 

\item \textbf{Extensionality}. If $\vDash \varphi \leftrightarrow \psi$ then $\vDash [\div\varphi]\chi\leftrightarrow [\div\psi]\chi$. 
\item  \textbf{Conjunctive Inclusion}. $[\div (\varphi\wedge\psi)]\neg B_a\varphi\ra ([\div (\varphi\wedge\psi)] B_a\chi \ra [\div \varphi] B_a\chi)$. 
\item \textbf{Conjunctive Overlap}. $([\div \varphi] B_a\chi \wedge [\div \psi] B_a\chi)\rightarrow [\div (\varphi\wedge \psi)] B_a\chi $. 
\item \textbf{Consistency}. $\neg B_a\bot\ra [\div\varphi] \neg B_a\bot$. \end{itemize}

\noindent From the list above we omitted the axiom called \textbf{Recovery} (expansion by $\varphi$ after contraction by $\varphi$ restores the agent's original beliefs), as it involves an expansion operation which we do not focus on.\footnote{Notice, however, that \textbf{Recovery} does not hold in general in this framework: while the contraction operator $[\div\varphi]$ can informally be viewed as a dual of the public announcement operator $[!\varphi]$, it is not a formal dual, in the sense that contraction by $\varphi$ followed by expansion by $\varphi$ does not in general recover the original model (whether we model expansion by world- or arrow-eliminations). We explore this issue in a longer version of the paper.} We added to this list a natural property called \textbf{Consistency}, which is discussed in \cite{segerberg1999two,herzig2017dynamic}.

 Some of these properties are valid in \textsf{HPAL}. In particular, \textbf{Closure} and \textbf{Extensionality} hold unrestrictedly in our semantic framework. \textbf{Consistency} also holds, as the contraction operation never removes edges. For the other properties, the situation is more mixed, as the next two subsections show.

\subsection{Success}
Recall that $\varphi$ is successful in \textsf{PAL} if it is believed after it is announced: $[!\varphi]B_a\varphi$ is valid for all $a\in Ag$. A natural analogue of this notion in our contraction settings is that $\varphi$ is \textit{not believed} after a hedged announcement that it might be false, viz.~$[\div\varphi]\neg B_a\varphi$ for all $a\in Ag$. The only caveat concerns the case where $\varphi$ is valid in the model. In that case, since there are no $\neg\varphi$-possibilities, contracting on $\varphi$ leaves the model unchanged and so may leave agents' beliefs in $\varphi$ unchanged.\footnote{This mirrors the AGM constraint that beliefs in validities cannot be given up \cite{alchurron1985logic,sep-logic-belief-revision}.}  This illustrates why \textbf{Success} is given by the conditional $\exists\neg \varphi\rightarrow [\div\varphi]\neg B_a\varphi$. Say that a formula $\varphi$ is \textit{successful} if it satisfies this property.
 
  As with preservation, not all formulas are successful. For instance, the hedged announcement that $p\wedge B_a\neg p$ might be true in Figure~\ref{fig:moore} is unsuccessful: after updating on Bob's announcement that she might have a false belief that $\neg p$, by contracting on $\neg (p\wedge B_a\neg p)$, Alice will come to believe that she believes neither $p$ nor $\neg p$, and so she will believe that she doesn't have a false belief that $p$ (i.e. $\M^{\div\neg (p\wedge B_a\neg p)}\vDash B_a\neg (p\wedge B_a\neg p)$). 
In \textsf{PAL}, if a formula remains true after being announced,
then it is believed after the announcement \cite[Proposition 4.32]{van2008dynamic}. A similar result holds for \textsf{HPAL}. 

    \begin{proposition}\label{prop:preserved->successful}
In \textsf{HPAL}, if $\neg\varphi\in\Lang^{\div}$ is preserved, then $\varphi$ is successful.
    \end{proposition}
  \begin{proof} Let $\M=(W,R,V)$ and suppose $\M,w\vDash \exists\neg\varphi$.
Assume that $\neg\varphi$ is preserved. We show that
$\M,w\vDash[\div\varphi]\neg B_a\varphi$. There are two cases. First, suppose that $\M, w\vDash \neg B_a\varphi$. Then, there is $v\in R_a(w)=R^{\div\varphi}_a(w)$ such that $\M, v\vDash\neg \varphi$. As $\neg\varphi$ is preserved, $\M, v\vDash [\div\varphi]\neg\varphi$ and so $\M, w\vDash [\div\varphi]\neg B_a\varphi$. Second, suppose $\M,w\vDash B_a\varphi$. Since $\M, w\vDash\exists\neg\varphi$, there is a $v\in W$ with $\M, v\vDash\neg\varphi$ and $v\in R_a^{\div\varphi}(w)$. By preservation of $\neg\varphi$, $\M, v\vDash [\div\varphi]\neg\varphi$ and so $\M, w\vDash [\div\varphi]\neg B_a\varphi$. 
  \end{proof}
\noindent Given that existential formulas are preserved (Theorem~\ref{thm:existential}), we then obtain the following corollary, which gives a sufficient condition for a formula to be successful:
\begin{corollary}\label{cor:successhpal}
   If $\varphi\in\Lang^\div$ is equivalent in \textsf{HPAL} to a universal formula, then $\varphi$ is successful in \textsf{HPAL}.
\end{corollary}
\begin{proof}
If $\varphi$ is equivalent to a universal formula, $\neg\varphi$ is
equivalent to an existential formula. By Theorem~\ref{thm:existential} existential formulas are preserved, and by Prop.~\ref{prop:preserved->successful} if $\neg\varphi$ is preserved, then $\varphi$ is successful.
\end{proof}

\subsection{Other properties}
The previous section showed that \textbf{Success} is not valid in \textsf{HPAL}. However, by Corollary \ref{cor:successhpal}, all universal formulas are successful and in particular all propositional formulas are successful. The same holds for \textbf{Inclusion}, \textbf{Conjunctive Inclusion} and \textbf{Conjunctive Overlap}.
          These postulates fail for reasons similar to the failure of \textbf{Success}. For instance, in Figure~\ref{fig:contraction1}, 
after Bob's hedged announcement, Alice comes to believe that she doesn't believe $p$---a belief that she didn't have before, which is a failure of \textbf{Inclusion} ($\M, w\vDash \neg B_a\neg B_a\neg p\wedge [\div\neg p]B_a\neg B_a\neg p$). 
This is to be expected: on the dynamic approach, the truths of doxastic formulas can be influenced by hedged announcements. However, all three postulates hold if the believed formula is propositional. Let $\mathfrak{K}$ be the class of all Kripke models. 
\begin{proposition} The following are valid in \textsf{HPAL}, where $\chi$ and $\lambda$ are propositional and $\varphi,\psi\in\Lang^{\div}$:
\begin{enumerate}
    \item  \textbf{Propositional Inclusion}. $[\div\varphi]B_a \chi \rightarrow B_a \chi$.
    \item  \textbf{Propositional Conjunctive Inclusion}. 
        $ [\div \lambda\wedge\psi]\neg B_a\lambda\ra ([\div \lambda\wedge\psi] B_a\chi \ra [\div \lambda] B_a\chi)$.
    \item \textbf{Propositional Conjunctive Overlap}. 
        $ ([\div \varphi] B_a\chi \wedge [\div \psi] B_a\chi)\rightarrow [\div \varphi\wedge \psi] B_a\chi$.
\end{enumerate}
\end{proposition}
\begin{proof}[Proof sketch.]
    \textbf{Propositional Inclusion} follows from $R_a(w)\subseteq R_a^{\div\varphi}(w)$ and from propositional truth being unaffected by contraction. \textbf{Propositional Conjunctive Inclusion} and \textbf{Overlap} follow because contracting on a conjunction introduces
no alternatives beyond those required to give up one of its conjuncts.
\end{proof}

\noindent Finally, while \textbf{Vacuity} is also among the properties that do not hold unrestrictedly, the following strong version holds:
\begin{itemize}
    \item \textbf{Strong Vacuity}.     $\bigwedge_{a\in Ag}\forall\neg B_a\varphi \rightarrow (\chi\leftrightarrow [\div\varphi]\chi)$.
\end{itemize}
This principle strengthens the standard \textbf{Vacuity} principle listed above by requiring that \textit{no agent} believes the contracted formula \textit{at any world in the model}. Such restrictions reflect the multi-agent and public nature of our contraction operation. They are not required in classic AGM theory, where beliefs     are propositional, nor in dynamic doxastic logic, which is single-agent \cite{alchurron1985logic,segerberg1999two}.
\begin{proposition}  \textbf{Strong Vacuity} is a theorem of \textsf{HPAL}.    
  \end{proposition}
      \begin{proof}Let $\varphi, \chi\in\Lang^\div$ and consider a Kripke model $\M=(W,R,V)$ and a world $w$ in it.     Assume that $\M,w\vDash \forall\neg B_a\varphi$, for all $a\in Ag$, i.e., $\M,u\vDash \neg B_a\varphi$ for all $u\in W$ and $a\in Ag$. Then $R_a(u)=R^{\div\varphi}_a(u)$ for all worlds $u$     and all agents $a$, i.e., $\M=\M^{\div\varphi}$ and so $\M,w\vDash\chi$ iff $\M^{\div\varphi},w\vDash\chi$, that is, $\M,w\vDash\chi$ iff $\M,w\vDash[\div\varphi]\chi$.
 	
	    \end{proof}

\section{Generalized \textsf{DEL}}
In this section, we introduce a generalized DEL framework that can model belief contraction resulting
from different kinds of announcements, like hedged private announcements.

Let the \textit{\textsf{DEL} language} $\Lang_{DEL}$ be the language generated by the grammar of the doxastic language $\Lang$ extended with the following clauses $\varphi::=[\E,e]\varphi\mid \forall \varphi$, where $\E$ is a generalized event model (introduced below) and $e$ is an event in it. The formula $[\E,e]\varphi$ reads ``after $(\E,e)$ happens, $\varphi$ is the case''. 
\begin{definition}[Generalized event model]
    A \emph{generalized event model} is a tuple $\E=(E, Q, Q^+, pre)$ where $E$ and $Q_a$ are defined as in standard event models, $Q^+_a\subseteq E\times E$ is an accessibility relation between events defined for all agents $a\in Ag$, and $pre: E\ra\Lang_{DEL}$ is a precondition function.
\end{definition}

A generalized event model is a standard event model extended with additional accessibility relations $Q^+_a$, for each agent $a\in Ag$,  and in which preconditions can be formulas of the dynamic language $\Lang_{DEL}$. As before, we use notation $ Q_a^+(e)=\{f\in E\colon (e,f)\in Q^+_a\}$. Intuitively, both $Q_a(e)$ and $Q^+_a(e)$ represent events that are accessible, and thus considered possible, by agent $a$ at event $e$. The difference, which will become formally clear with the product update definition below, is that 
$Q^+_a(e)$ contains alternatives that are \textit{introduced by event $e$ as possible for agent $a$} 
(irrespective of $a$'s prior beliefs), while $Q_a(e)$ is the standard accessibility relation containing possibilities that are not \textit{ruled out for agent $a$ by event $e$}. 

The distinction is illustrated by the event model $\E$ in Figure~\ref{fig:ProductUpdatePublic}, which is the event model corresponding to Example
\ref{eg:mightopen}. Intuitively, $f$ represents the event \textit{Bob tells Alice her office door might be open and it is actually closed}, and $h$ represents the event \textit{Bob tells Alice her office door might be open and it is actually open}. For Alice, Bob's announcement doesn't rule out either $f$ or $h$, but it introduces $h$ as possible. Formally, this means $Q_a(f)=Q_a(h)=\{f,h\}$, while $Q_a^+(f)=Q_a^+(h)=\{h\}$.

\begin{figure}[H]
\centering
\begin{minipage}{0.33\textwidth}
    \centering    \begin{tikzpicture}
         \node[label={[label distance=-4mm]left:$\M$}] (A) {\begin{tikzpicture}[node distance=2cm, ->, >=Stealth, thick]
        \tikzstyle{state} = [circle, draw, thick, minimum size=0.6cm, inner sep=0pt, align=center]
            \node[state, label=below:$w$] (w) {\(\neg p\)};
  \node[state, right of=w, label=below:$v$] (v) {\(p\)};
   \path[->](v) edge node[above] {$a$} (w)
(w) edge[loop left] node[above,xshift=-3pt] {$a$} (w);
        \end{tikzpicture}};
    \end{tikzpicture}
\end{minipage}
\begin{minipage}{0.30\textwidth}
    \centering    \begin{tikzpicture}
         \node[label=left:$\E$] (A) {\begin{tikzpicture}[node distance=2cm, ->, >=Stealth, thick]
        \tikzstyle{event} = [rectangle, draw, minimum size=0.6cm, inner sep=0.4pt, align=center, thick]
            \node[event, label=below:$f$] (f) {\(\neg p\)};
            \node[event, label=below:$h$, right of=f] (h) {\(p\)};
\path[->] 
(h) edge[loop right] node[above] {$a+$} (h)
(f) edge node[above] {$a+$} (h);

        \end{tikzpicture}};
    \end{tikzpicture}
\end{minipage}
\begin{minipage}{0.35\textwidth}
    \centering
\begin{tikzpicture}
         \node[label=left:$\M\otimes\E$] (C) {\begin{tikzpicture}[node distance=2cm, ->, >=Stealth, thick]
        \tikzstyle{state} = [circle, draw, thick, minimum size=0.6cm, inner sep=0pt, align=center]
  \node[state, label=below:${(w,f)}$] (w2) {\(\neg p\)};
  \node[state, right of=w2, label=below:${(v,h)}$] (v2) {\( p\)};
   \path[<->]
  (v2) edge node[above] {$a$} (w2)
        (w2) edge[loop left] node[below] {$a$} (w2)
        (v2) edge[loop right] node[above] {$a$} (v2);
        \end{tikzpicture}};
    \end{tikzpicture}
\end{minipage}
\vspace{-10pt}
        \caption{\textbf{Left}: The Kripke model $\M$ representing Alice's initial belief states.                 \textbf{Middle}: The generalized event model $\E=(E,Q,Q^+,pre)$ representing Bob's announcement to Alice that $p$ might hold. An edge from an event $f$ to an event $h$ labeled by $+a$ means that $h\in Q_a^+(f)$. We omit the $Q$-edges since $Q_a(e)=E$ for all $e\in E$.          \textbf{Right}: The Kripke model $\M\otimes\E=(W^\E,R^\E,V^\E)$ representing Alice's belief states after Bob's announcement to Alice.}\label{fig:ProductUpdatePublic}
\end{figure}
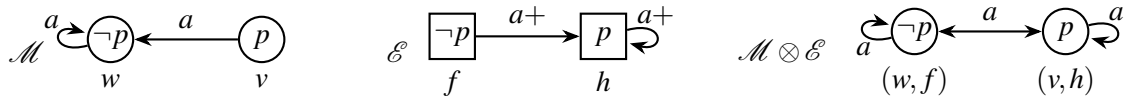

Given this interpretation, the product update is then defined as the following:
\begin{definition}[Generalized product update]
Given a Kripke model $\M=(W, R, V)$ and a generalized event model $\E=(E, Q, Q^+, pre)$, their product update $\M\otimes\E$ is defined as $\M\otimes\E=(W^\E, R^\E, V^\E)$ where $W^\E$ and $V^\E$ are defined as in standard product updates (cf. Section~\ref{sec:standardDEL}) and $R^\E$ is given by:
    $$(v,f)\in R_a^\E(w,e) \text{ iff either }f\in Q_a^+(e),\text{ or } v\in R_a(w) \text{ and } f\in Q_a(e).$$
\end{definition}
\noindent This product update behaves like a standard product update in eliminating possibilities (cf. Section~\ref{sec:standardDEL}), but additionally allows to expand the set of accessible worlds by adding all worlds $(v,f)$ where $f$ is a newly considered possibility. In particular, the first disjunct allows agents to start considering worlds paired with events in $Q_a^+$, regardless of the original accessibility relation.

The semantics of the dynamic modality is standard \cite{van2008dynamic}:
\begin{table}[H]
	   \centering
            \begin{tabular}{lll}
				$\M, w\vDash [\E,e]\varphi$ & iff & if $\M,w\vDash pre(e)$ then $\M\otimes \E, (w,e)\vDash\varphi$.\\
            \end{tabular}
	\end{table}

\vspace{-10pt}
The following definition shows that generalized \textsf{DEL} can express contraction as induced by hedged public announcements.	\begin{definition}[Event model for contraction on $\varphi$] Let $\varphi\in \Lang_{DEL}$.
		An \emph{event model for contraction on $\varphi$} is a generalized event model $\E(\div\varphi)=(E,Q, Q^+, pre)$ where 
\begin{enumerate}
    \item $E=\{\varphi\wedge \bigwedge_{a\in A} \neg B_a\varphi \wedge \bigwedge_{a\in Ag\setminus A}B_a\varphi
								\colon A\subseteq Ag\}\cup \{\neg\varphi\wedge \bigwedge_{a\in A} \neg B_a\varphi \wedge \bigwedge_{a\in Ag\setminus A}B_a\varphi
								\colon A\subseteq Ag\}.$
     \item For all $a\in Ag$ and $e\in E$,
            $Q_a(e)=E$ and
        $Q^+_a(e)=\begin{cases}
				\{f\in E\colon pre(f)\vDash \neg\varphi\} & \text{ if } pre(e)\vDash B_a\varphi;\\
				\emptyset & \text{ otherwise}.
			\end{cases}$
    \item For all $e\in E$, $pre(e)=e$.
\end{enumerate}
	\end{definition}
\noindent Event model $\E(\div\varphi)$ contains events for all possible configurations of the truth value of $\varphi$ and which agents believe $\varphi$. The relation $Q_a$ is universal, capturing that the announcement reveals nothing about which event occurred. The relation $Q_a^+$ captures contraction: agents who believed $\varphi$ come to consider $\neg\varphi$-events possible, while the others are unaffected.
\begin{proposition}\label{prop:contractionisomorphism}
		For any Kripke model $\M$ and formula $\varphi\in\Lang$, $\M^{\div\varphi}$ is isomorphic to $\M\otimes\E(\div\varphi)$.
	\end{proposition}
    \begin{proof}[Proof sketch.] Both updates preserve all worlds of $\M$, so there is an isomorphism between the updated models: as $R_a(w)\subseteq R_a^{\div\varphi}(w)$ and $Q_a(e)=E$ for all $a\in Ag,e\in E$, all relations in $\M$ are preserved by both updates. Also, the only relations added in both cases are from $B_a\varphi$-worlds to $\neg\varphi$-worlds.
    \end{proof}

Next we consider an example involving hedged private announcement:
\begin{example}\label{eg:private}
Alice just told Clark that she locked her office door before she left. Later, while Clark is evidently distracted, Bob privately tells Alice that her office door might actually be open. Alice comes to suspend judgment on whether she locked the door, while Clark does not notice this exchange, continuing to believe that Alice's office door is closed and that Alice believes that. 
\end{example}

We model Bob's private announcement using the generalized event model $\E$ as described below:
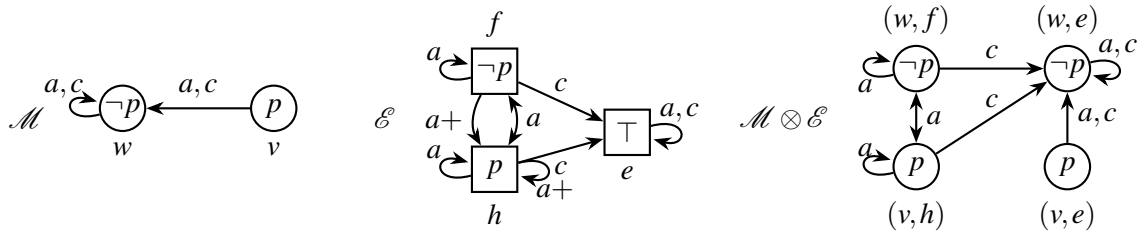
\begin{figure}[H]
\centering
\begin{minipage}{0.33\textwidth}
    \centering    \begin{tikzpicture}
         \node[label={[label distance=-4mm]left:$\M$}] (A) {\begin{tikzpicture}[node distance=2cm, ->, >=Stealth, thick]
        \tikzstyle{state} = [circle, draw, thick, minimum size=0.6cm, inner sep=0pt, align=center]
            \node[state, label=below:$w$] (w) {\(\neg p\)};
  \node[state, right of=w, label=below:$v$] (v) {\(p\)};
   \path[->](v) edge node[above] {$a,c$} (w)
(w) edge[loop left] node[above,xshift=-3pt] {$a,c$} (w);
        \end{tikzpicture}};
    \end{tikzpicture}
\end{minipage}
\begin{minipage}{0.30\textwidth}
    \centering    \begin{tikzpicture}
         \node[label=left:$\E$] (A) {\begin{tikzpicture}[node distance=2cm, ->, >=Stealth, thick]
        \tikzstyle{event} = [rectangle, draw, minimum size=0.6cm, inner sep=0.4pt, align=center, thick]
            \node[event, label=below:$e$] (e) {\(\top\)};
            \node[event, label=above:$f$, left of=e,yshift=23pt,xshift=7pt] (f) {\(\neg p\)};
            \node[event, label=below:$h$, below of=f,yshift=20pt] (h) {\(p\)};
\path[->] (f) edge node[above] {$c$} (e)
(h) edge node[below] {$c$} (e)
(e) edge[loop right] node[above] {$a,c$} (e)
(f) edge[loop left] node[above,xshift=-3pt] {$a$} (f)
(h) edge[loop left] node[above, ,xshift=-3pt] {$a$} (h);

\path[->](f) edge[bend right] node[left] {$a+$} (h);
\path[<->](h) edge[bend right] node[right] {$a$} (f);
\path[->] (h) edge[loop right] node[below, xshift=3pt] {$a+$} (h);

        \end{tikzpicture}};
    \end{tikzpicture}
\end{minipage}
\begin{minipage}{0.35\textwidth}
    \centering
\begin{tikzpicture}
         \node[label=left:$\M\otimes\E$] (C) {\begin{tikzpicture}[node distance=2cm, ->, >=Stealth, thick]
        \tikzstyle{state} = [circle, draw, thick, minimum size=0.6cm, inner sep=0pt, align=center]
  \node[state, label=above:${(w,f)}$] (w2) {\(\neg p\)};
  \node[state, below of=w, label=below:${(v,h)}$, yshift=20pt] (v2) {\( p\)};
\node[state, right of=w, label=above:${(w,e)}$] (w) {\(\neg p\)};
  \node[state, below of=w, label=below:${(v,e)}$, yshift=20pt] (v) {\( p\)};
   \path[->]
  (v) edge node[right] {$a,c$} (w)
(w) edge[loop right] node[above] {$a,c$} (w);
   \path[<->]
  (v2) edge node[right] {$a$} (w2);
   \path[->]
  (v2) edge node[above] {$c$} (w)
      (w2) edge node[above] {$c$} (w);
           \path[->]
        (w2) edge[loop left] node[below] {$a$} (w2)
        (v2) edge[loop left] node[above] {$a$} (v2);
        \end{tikzpicture}};
    \end{tikzpicture}
\end{minipage}
\vspace{-10pt}
        \caption{\textbf{Left}: The Kripke model $\M$ representing Alice and Clark's initial belief states. We have $\M\vDash B_a \neg p\wedge B_c \neg p \wedge B_c B_a\neg p$. 
                \textbf{Middle}: The generalized event model $\E=(E,Q,Q^+,pre)$ representing Bob's private announcement to Alice that $p$ might hold. An edge from an event $f$ to an event $h$ labeled by $+a$ means that $h\in Q_a^+(f)$. If an edge's label does not contain $+$, then the edge is a $Q$-relation, e.g. $e\in Q_c(f)$.
        \textbf{Right}: The Kripke model $\M\otimes\E=(W^\E,R^\E,V^\E)$ representing Alice and Clark's belief states after Bob's private announcement to Alice.}
            \label{fig:private}
\end{figure}

    \noindent We now look into what is the logic of generalized updates. First, we define the composition of two generalized event models and then show that it is equivalent to the sequential updates of the two. 
 \begin{definition}[Composition]
	Given two generalized event models $\E=(E^\E, Q^\E, Q^{+\E}, pre^\E)$ and $\F=(E^\F,  Q^\F, Q^{+\F}, pre^\F)$, then their composition is $\E\circ \F=(E,Q, Q^{+}, pre)$ where
	\begin{itemize}
    \vspace{-5pt}
      \setlength\itemsep{0em}
		\item $E=E^\E\times E^\F$;
		\item $(e',f')\in Q_a((e,f))$ iff $e'\in Q_a^\E(e)$ and $f'\in Q_a^\F(f)$;
		\item $(e', f')\in Q^{+}_a((e,f))$ iff $f'\in Q^{+\F}_a(f)$, or $f'\in Q_a^\F(f)$ and $e'\in Q^{+\E}_a(e)$;
	\item $pre((e,f))=pre^\E(e)\wedge [\E, e]pre^\F(f)$.
	\end{itemize}
\end{definition}

    The following shows that this is given by the system presented in Table~\ref{tab:GDEL}.\footnote{Note that, unlike the axiomatization of \textsf{HPAL}, the axiomatization of \textsf{GDEL} includes a composition axiom (A6) but without RE as a valid rule of inference. We conjecture that we can omit A6 and include RE as a valid inference rule, though we leave the verification of this conjecture to future work.}
	
		\begin{table}
			\centering
            \small
			\begin{tabular}{cc}
               				                                                				
												                CL & all classical propositional tautologies\\
				$\forall$K & $\forall(\varphi\ra\psi)\ra (\forall\varphi\ra\forall\psi)$\\
                $\forall$T & $\forall\varphi\ra \varphi$\\
                $\forall$5 & $\neg\forall\varphi\ra \forall\neg\forall\varphi$\\
                $\forall$B & $\forall\varphi\ra B_a\varphi$\\
				BK & $B_a(\varphi\ra\psi)\ra (B_a\varphi\ra B_a\psi)$\\
                  \midrule
				A1 & $[\E, e]p\leftrightarrow (pre(e)\ra p)$\\
				A2 & $[\E,e]\neg\varphi\leftrightarrow(pre(e)\ra \neg [\E, e]\varphi)$\\
				A3 & $[\E, e](\varphi\wedge\psi)\leftrightarrow ([\E,e]\varphi\wedge [\E,e]\psi)$\\
				A4& $[\E, e]\forall\varphi\leftrightarrow (pre(e)\ra\bigwedge_{f\in E}\forall[\E,f]\varphi)$\\
				A5 & $[\E,e]B_a\varphi\leftrightarrow (pre(e)\ra (\bigwedge_{f\in Q_a^+(e)}\forall [\E,f]\varphi\wedge  \bigwedge_{f\in Q_a(e)}B_a[\E, f]\varphi))$\\
				A6 & $[\E, e][\F, f]\varphi\leftrightarrow [\E\circ\F, (e,f)]\varphi$\\
                				\multicolumn{2}{c}{MP and necessitation rules for all modalities}  \\
				\hline
			\end{tabular}
			\caption{\small The Theory \textsf{GDEL}}
			\label{tab:GDEL}
		\end{table}
     
		\begin{theorem}
			The dynamic logic of generalized updates is completely axiomatized by \textsf{GDEL}.
		\end{theorem} 
        \begin{proof}
[Proof sketch.] The proof strategy is analogous to that used for Theorem~\ref{thm:HPALcompleteness}, using the composition axiom A6 instead of the inference rule RE.
        \end{proof}
        
        \noindent The theory \textsf{GDEL} contains the static axioms of \textsf{HPAL}, together with the inference rules for all modalities. Additionally, it contains reduction axioms for event models analogous to those of standard \textsf{DEL}.  Axiom A5 captures the effect of generalized updates on belief via the relations $Q_a$ and $Q_a^+$.

 The next theorem shows that generalized updates always produce a model that is a simulation of a refinement of the initial model (refinement and simulation are defined standardly \cite{bluebook}, but see Appendix). As noted in \cite{bozzelli2014refinement,ditmarsch2025simulation}, refinements decrease uncertainty by eliminating possibilities, while simulations capture increases in uncertainty by adding possibilities. Generalized updates thus can be seen as updates where agents may rule out possibilities while starting to consider new ones, as in belief revision.

	\begin{theorem}
		For any Kripke model $\M$ and generalized event model $\E$, $\M\otimes\E$ is a simulation of a refinement of $\M$.
	\end{theorem}
	\begin{proof}[Proof sketch.]The idea is to split a generalized update into two steps. First ignore the $Q_a^+$ edges and obtain a standard \textsf{DEL} update, which gives a refinement of the initial model \cite{bozzelli2014refinement}. The full generalized update then only adds edges to this model, so it is a simulation of that refinement.
	\end{proof}

Before concluding this section, we discuss one limitation of the contraction operation as defined in Section \ref{sec:4contraction}, and how the generalized product update helps to overcome it.
\begin{example}
    Alice and Clark walk down the hallway. Alice believes that she closed her office door ($\neg p$) and that there will be a department meeting in the afternoon ($q$). Clark, however, is uncertain about both questions, though he believes that Alice believes the truth no matter what it is, and so does Alice. They then encounter Bob, who tells them that Alice's office door might be open. 
\end{example}
Intuitively, given Bob's announcement, Clark will come to think that, if Alice initially believes that her office door is closed, then she'll come to suspend judgment about whether her office door is open, but she will remain confident about whether there will be a department meeting in the afternoon. So Alice's and Clark's beliefs should be modeled by $\N$ in \Cref{fig:hansexample}. \begin{figure}[H]
\centering
\begin{minipage}{0.4\textwidth}
    \centering    \begin{tikzpicture}
         \node[label={[label distance=-4mm]left:$\M$}] (A) {\begin{tikzpicture}[node distance=2cm, ->, >=Stealth, thick]
        \tikzstyle{state} = [circle, draw, thick, minimum size=0.95cm, inner sep=0pt, align=center]
            \node[state, label=above:$u$] (w) {\(pq\)};
  \node[state, right of=w, label=above:$z$] (v) {\(p\neg q\)};
  \node[state, below of=w, label=below:$w$] (u) {\(\neg pq\)};
  \node[state, right of=u, label=below:$v$] (z) {\(\neg p\neg q\)};
  \path[->] (w) edge[loop left] node[left] {$a,c$} (w)
(v) edge[loop right] node[right] {$a,c$} (v)
(u) edge[loop left] node[left] {$a,c$}(u)
(z) edge[loop right] node[right] {$a,c$} (z);
\path[<->] (w) edge node[above] {$c$} (v)
  (w) edge node[left] {$c$} (u)
  (v) edge node[right] {$c$} (z)
   (u) edge node[above] {$c$} (z);
        \end{tikzpicture}};
    \end{tikzpicture}
\end{minipage}
\vspace{-10pt}
\begin{minipage}{0.4\textwidth}
    \centering
\begin{tikzpicture}
         \node[label={[label distance=-4mm]left:$\N$}] (A) {\begin{tikzpicture}[node distance=2cm, ->, >=Stealth, thick]
        \tikzstyle{state} = [circle, draw, thick, minimum size=0.95cm, inner sep=0pt, align=center]
            \node[state, label=above:$u$] (w) {\(pq\)};
  \node[state, right of=w, label=above:$z$] (v) {\(p\neg q\)};
  \node[state, below of=w, label=below:$w$] (u) {\(\neg pq\)};
  \node[state, right of=u, label=below:$v$] (z) {\(\neg p \neg q\)};
  \path[->] (w) edge (v)
  (u) edge (z)
  (w) edge[loop left] node[left] {$a,c$} (w)
(v) edge[loop right] node[right] {$a,c$} (v)
(u) edge[loop left] node[left] {$a,c$}  (u)
(z) edge[loop right] node[right] {$a,c$}  (z)
   (u) edge[bend left, red] node[left]{$a$} (w)
(z) edge[bend right, red] node[right]{$a$} (v);
\path[<->] (w) edge node[above] {$c$} (v)
  (w) edge node[left] {$c$} (u)
  (v) edge node[right] {$c$} (z)
   (u) edge node[above] {$c$} (z);
        \end{tikzpicture}};
    \end{tikzpicture}
\end{minipage}
        \caption{\textbf{Left}: The Kripke model $\M$ representing Alice and Clark's initial belief states.         We have $\M\vDash B_c(B_ap\vee B_a\neg p) \wedge B_c(B_aq\vee B_a\neg q)$. \textbf{Right}: The Kripke model $\N$ representing Alice and Clark's belief states after Bob's hedged announcement, where $\N,w\vDash \neg B_c(B_a p\vee B_a\neg p)\wedge B_c(B_aq\vee B_a\neg q)$.}
    \label{fig:hansexample}
\end{figure}
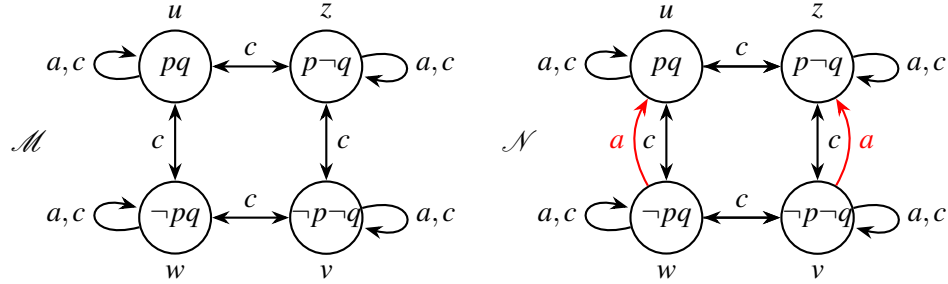
\begin{proposition}
    There is no formula $\varphi\in\Lang^{\div}$ such that $\M^{\div\varphi}\equiv \N$.
\end{proposition}
\begin{proof}
     Suppose towards a contradiction that $\M^{\div\varphi}\equiv\N$ for some $\varphi\in\Lang^{\div}$. Since every world in $\M$ satisfies a distinct set of propositional formulas and contraction preserves propositional valuation, it follows that for all $x\in W$, $(\M^{\div\varphi}, x)\equiv (\N, x)$. Since $u\not\in R_a(w)$ but $u\in R^{\div\varphi}_a(w)$, we have $\M, w\vDash B_a\varphi$ and $\M, u\vDash \neg\varphi$. Similarly, $\M, v\vDash B_a\varphi$ and $\M,z\vDash \neg\varphi$. By the definition of contraction, $\{x\in W: \M, x\vDash\neg\varphi\}\subseteq R_a^{\div\varphi}(w)$. So $z\in R_a^{\div\varphi}(w)$. So $\M^{\div\varphi}, w\vDash \neg B_aq$. Since $\N, w\vDash B_aq$, this contradicts the assumption that $(\M^{\div\varphi}, w)$ and $(\N, w)$ are modally equivalent. \end{proof}
One way of addressing this problem within the \textsf{HPAL} framework is to restrict the set of worlds that Alice starts considering to a subset of $p$-possibilities, e.g.~those that are compatible with her background knowledge or ``entrenched'' beliefs. We can also model these dynamics using \textsf{GDEL}. Consider $\E=(E, Q, Q^+, pre)$ depicted as in Figure \ref{fig:EventHans}, which captures three aspects of the update: (i) Bob's announcement doesn't eliminate any possibilities (so the $Q$-relations are universal); (ii) if Alice believes $p$, then Bob's announcement doesn't make her consider new possibilities ($Q_a^+(e)=Q_a^+(g)=\emptyset$) (iii) if Alice believes $\neg p$, then Bob's announcement makes her consider $p$ possible without revising her beliefs about $q$ ($Q_a^+(f)=\{e\}$ and $Q_a^+(h)=\{g\}$). It's not hard to check that $\M\otimes\E=\N$. 
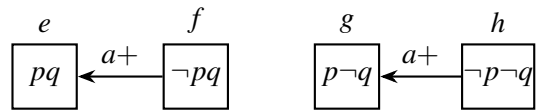
\begin{figure}[H]
    \centering
        \centering   
        \begin{tikzpicture}[node distance=2cm, ->, >=Stealth, thick]
        \tikzstyle{event} = [rectangle, draw, minimum size=0.85cm, inner sep=0.4pt, align=center, thick]
            \node[event, label=above:$e$] (e) {\(pq\)};
\node[event, label=above:$f$, right of=e] (f) {\(\neg pq\)};
            \node[event, label=above:$g$, right of=f] (g) {\(p\neg q\)};
            \node[event, label=above:$h$, right of=g] (h) {\(\neg p\neg q\)};
\path[->] (f) edge node[above] {$a+$} (e)
(h) edge node[above] {$a+$} (g);
    \end{tikzpicture}
    \caption{    
    The generalized event model $\E=(E,Q,Q^+,pre)$ representing Bob's announcement to Alice and Clark. As before, we omit the $Q$-edges, since $Q_i(x)=E$ for all events $x\in E$ and agents $i\in \{a,c\}$.  Each event's precondition is given by the conjunction of the literals listed inside the event. So for example event $e$ has precondition $p\wedge q$.}\label{fig:EventHans}
\end{figure}
\section{Conclusion and future work}

In this paper, we introduced a logic for belief contraction due to hedged public announcements, and then a generalized \textsf{DEL} theory for belief expansion and contraction. There are many questions left open for future work, such as: the full characterization of successful formulas for \textsf{HPAL}, the interaction between contraction and expansion operators, properties of a revision operator defined in terms of their combinations via the Levi identity \cite{sep-logic-belief-revision}, and the iterative properties of \textsf{GDEL}. Of special importance are the closure properties of \textsf{GDEL}. While updates with standard event models and plausibility event models preserve natural properties of accessibility relations, such as transitivity and Euclideanness, this is not true of generalized product updates, even if both the $Q$- and the $Q^+$-relations in the generalized event model are equivalence relations. This has the undesirable consequence that an agent who had full introspection of her own beliefs prior to the update may have higher-order uncertainties about her own beliefs after the update. While it is possible to enforce preservation of introspection in our framework, it is an open question for which class of generalized event models this can be guaranteed.

In addition to the above questions, we also plan to investigate how our model compares with alternative models of belief contraction \cite{leitgeb2007dynamic, girard2012general} as well as other models of information loss, such as forgetting and awareness growth \cite{ditmarsch2009awareness, benthem2010awareness}, and the dynamic approach to interpreting conditionals with modal antecedents. In particular, one may ask how expressive our generalized updates are with respect to plausibility updates. In one sense, we can transform any distinguished Kripke model $\mathcal{K}=(W, R, V)$ into any Kripke model $\mathcal{K'}=(W', R', V')$ with $W=W'$ and $V=V'$ by adding the edges in $\mathcal{K'}$ without keeping any edges in $\mathcal{K}$.\footnote{Recall that $\mathcal{K}=(W, R, V)$ is \textit{distinguished} if for every world $w$ there is a unique formula $\varphi$ such that $\M, w\vDash \varphi$. The simulation strategy is similar to the strategy used in \cite{ditmarsch2008semantic} to show that in \textsf{DEL} with postconditions, one can transform any Kripke model in any other Kripke model. The strategy there is also to get rid of the structure of the initial Kripke model and then reconstruct it by means of carefully devised event models with postconditions.} To this extent, our framework can emulate the belief dynamics that can be modeled in the plausibility framework, though we'll leave a detailed comparison to future works.

\subsubsection*{Acknowledgments}
We would like to acknowledge Hans van Ditmarsch for suggesting the paper's title and for very helpful comments on the paper, in particular for noticing that the axiomatization of \textsf{HPAL} was missing replacement of equivalents. Gaia would also like to acknowledge Thomas Bolander for very helpful initial discussions on this work and contraction in DEL. We also thank three anonymous reviewers for very helpful comments. Gaia Belardinelli is funded by Independent Research Fund Denmark (grant no.
4255-00020B).
\appendix
\section{Proofs}
Throughout, by \textit{logical equivalence} we mean provable in $\mathbf{K}$. We let $\llbracket\varphi\rrbracket=\{w\in W\colon\M,w\vDash \varphi\}$, and $\llbracket\varphi\rrbracket^{\div\psi}=\{w\in W^{\div\psi}\colon\M^{\div\psi},w\vDash \varphi\}$.

\bigskip
	\begin{definition}[Bisimulation, refinement, simulation]\label{def:bisimulation}
		Let $\M=(W,R,V)$ and $\M'=(W',R',V')$ be two Kripke models. A \emph{bisimulation} between $\M$ and $\M'$ 
		is a non-empty relation $Z\subseteq W \times W'$ such that for all $(w,w')\in Z$ and $a\in Ag$:
		
		\noindent \begin{itemize}
			\item[] (Atom): $w\in V(p)$ iff $w'\in V'(p)$, for all $p\in At$;
			\item[] (Forth): If $v\in R_a(w)$ then there exists $v'\in W'$ such that $v'\in R'_a(w')$ and $(v,v')\in Z$;
			\item[] (Back): If $v'\in R'_a(w')$ then there exists $v\in W$ such that $v\in R_a(w)$ and $(v,v')\in Z$;
		\end{itemize}
		When a bisimulation exists between two models, we say that they are \emph{bisimilar}. 
		A relation that satisfies Atom and Back is a \emph{refinement}. When a refinement exists between $\M$ and $\M'$, we say that $\M'$ is a \emph{refinement} of $\M$. A relation that satisfies Atom and Forth is a \emph{simulation}. When a simulation exists between $\M$ and $\M'$, we say that $\M'$ is a \emph{simulation} of $\M$.
	\end{definition}

\bigskip

\noindent\textit{Proof of Theorem 4.2.}
			    For soundness, the only non-trivial axiom is A5, namely $[\div\varphi]B_a\psi\leftrightarrow (B_a[\div\varphi]\psi\wedge (B_a\varphi\rightarrow \forall (\neg\varphi\rightarrow[\div\varphi]\psi)))$. Suppose $\M, w\vDash [\div\varphi] B_a\psi$. We claim that $\M, w\vDash B_a[\div\varphi]\psi$. Let $v\in R_a(w)$. Then $v\in R_a^{\div\varphi}(w)$. By assumption, $\M^{\div\varphi},w\vDash B_a\psi$. So $\M^{\div\varphi}, v\vDash \psi$. Thus $\M, v\vDash [\div\varphi]\psi$. So $\M, w\vDash B_a[\div\varphi]\psi$. Next, we claim that, if $\M, w\vDash B_a\varphi$ and $\M, w\vDash[\div\varphi]B_a\psi$, then for every $v\in W$, $\M, v\vDash \neg\varphi\rightarrow [\div\varphi]\psi$. Let $v\in W$. Suppose $\M, v\vDash \neg\varphi$. Since $\M, w\vDash B_a\varphi$, $v\in R^{\div\varphi}(w)$. Since $\M, w\vDash [\div\varphi]B_a\psi$, $\M^{\div\varphi}, v\vDash \psi$. So $\M, v\vDash [\div\varphi]\psi$.

    Conversely, suppose $\M, w\vDash B_a[\div\varphi]\psi\wedge (B_a\varphi\rightarrow \forall (\neg\varphi\rightarrow [\div\varphi]\psi))$. There are two cases, either $\M, w\vDash \neg B_a\varphi$ or $\M, w\vDash B_a\varphi$. Suppose $\M, w\vDash \neg B_a\varphi$. Let $v\in R^{\div\varphi}(w)$. Then $v\in R_a(w)$. Since $\M, w\vDash B_a[\div\varphi]\psi$, we have $\M, v\vDash [\div\varphi]\psi$. So $\M^{\div\varphi}, v\vDash \psi$. Thus $\M, w\vDash [\div\varphi]B_a\psi$. Now suppose $\M, w\vDash B_a\varphi$. So $\M, w\vDash \forall (\neg\varphi\rightarrow [\div\varphi]\psi)$. Let $v\in R^{\div\varphi}(w)$. Either $v\in R_a(w)$ or $\M, v\vDash \neg\varphi$. In the first case, by the assumption that $\M, w\vDash B_a[\div\varphi]\psi$, we have $\M^{\div\varphi}, v\vDash \psi$. In the second case, since $\M, w\vDash \forall (\neg\varphi\rightarrow [\div\varphi]\psi)$, we have $\M, v\vDash [\div\varphi]\psi$ or equivalently $\M^{\div\varphi},v\vDash\psi$. So $\M, w\vDash [\div\varphi]B_a\psi$.

\smallskip
We now show that RE is validity preserving. Suppose $\vDash\varphi\leftrightarrow\psi$. We proceed by induction on the complexity of the environment $\chi$. If $\chi$ is a propositional variable $p\in At$, then $\chi[\varphi/p]=\varphi$ and $\chi[\psi/p]=\psi$. So $\vDash\chi[\varphi/p]\leftrightarrow \chi[\psi/p]$ by assumption. The cases of negation and conjunction follow from propositional logic. Suppose $\vDash \chi[\varphi/p]\leftrightarrow\chi[\psi/p]$. Then
\begin{table}[H]
    \centering
    \begin{tabular}{ccc}
      $\M, w\vDash B_a\chi[\varphi/p]$   & iff & for all $ v\in R_a(w)$ $\M,v\vDash \chi[\varphi/p]$  \\
         & iff & for all $ v\in R_a(w)$ $\M,v\vDash \chi[\psi/p]$  \\
         & iff &  $\M, w\vDash B_a\chi[\psi/p]$ 
    \end{tabular}
\end{table}
A similar argument shows that $\vDash \forall\chi[\varphi/p]\leftrightarrow\forall \chi[\psi/p]$. 
Let $\lambda\in\Lang^{\div}$. Then
\begin{table}[H]
    \centering
    \begin{tabular}{ccc}
      $\M, w\vDash [\div\lambda]\chi[\varphi/p]$   & iff & $\M^{\div\lambda},w\vDash \chi[\varphi/p]$  \\
         & iff & $\M^{\div\lambda},w\vDash \chi[\psi/p]$  \\
         & iff &  $\M, w\vDash [\div\lambda]\chi[\psi/p]$ 
    \end{tabular}
\end{table}
     Lastly, for any Kripke model $\M=(W, R, V)$, we claim that $\M^{\div\chi[\varphi/p]}=\M^{\div\chi[\psi/p]}$. It suffices to show that $R_a^{\div\chi[\varphi/p]}(w)=R_a^{\div\chi[\psi/p]}(w)$ for all $w\in W$ and $a\in Ag$. Suppose $\M, w\vDash \neg B_a\chi[\varphi/p]$. Then by the inductive hypothesis, $\M, w\vDash \neg B_a\chi[\psi/p]$ and so $R_a^{\div\chi[\varphi/p]}(w)=R_a(w)=R_a^{\div\chi[\psi/p]}(w)$. Suppose $\M, w\vDash B_a\chi[\varphi/p]$. Then by IH, $\M, w\vDash B_a\chi[\psi/p]$ and
\begin{table}[H]
    \centering
    \begin{tabular}{ccc}
      $R_a^{\div\chi[\varphi/p]}(w)$   & $=$ & $R_a(w)\cup\{v\in W: \M, v\vDash \neg \chi[\varphi/p]\}$\\
      & $=$ & $R_a(w)\cup\{v\in W: \M, v\vDash \neg \chi[\psi/p]\}$\\
      & $=$ & $R_a^{\div\chi[\psi/p]}(w)$
    \end{tabular}
\end{table}
So $\M^{\div\chi[\varphi/p]}=\M^{\div\chi[\psi/p]}$. Thus, for any Kripke model $\M$ and world $w$, 
\begin{table}[H]
    \centering
    \begin{tabular}{ccc}
   $\M, w\vDash [\div\chi[\varphi/p]]\lambda$ & iff &  $\M^{\div\chi[\varphi/p]}, w\vDash\lambda$  \\
         & iff &  $\M^{\div\chi[\psi/p]}, w\vDash\lambda$\\
         & iff &  $\M, w\vDash [\div\chi[\psi/p]]\lambda$
    \end{tabular}
\end{table}
     Given the reduction axioms, every formula with dynamic modalities is provably equivalent in      \textsf{HPAL} to a formula without dynamic modalities. So the completeness of      \textsf{HPAL} follows from the completeness of the basic doxastic logic with universal modalities.      \hfill$\Box$

 \bigskip
\noindent \textit{Proof of Theorem~5.5.}
 Suppose $Ag=\{a\}$. Let $\varphi\in\Lang$. Let $\M=(W, R, V)$ be a Kripke model that validates \textbf{S5}. 
 Suppose $\M, w\vDash \varphi$. Since $R_a$ is an equivalence relation, for all $v\in R_a(w)$, $\M, v\vDash \widetilde{B}_a\varphi$. Then $R_a(v)=R_a^{\div\neg\varphi}(v)$ for all $v\in R_a(w)$. Consider $Z=\{(v,v): v\in R_a(w)\}$.  We claim that $Z$ is a bisimulation between $(\M, w)$ and $(\M^{\div\neg\varphi}, w)$. Clearly the two worlds satisfy the same proposition atoms. Let $(v,v)\in Z$. Then $v\in R_a(w)$. Suppose $u\in R_a(v)$. Since $R_a$ is an equivalence relation, we have $u\in R_a(w)$ and so $(u,u)\in Z$. Moreover, since $R_a(v)=R_a(w)=R_a^{\div\neg\varphi}(w)=R_a^{\div\neg\varphi}(v)$, we have $u\in R_a^{\div\neg\varphi}(v)$. Similarly, suppose $u\in R_a^{\div\neg\varphi}(v)$. Then $u\in R_a^{\div\neg\varphi}(w)=R_a(w)=R_a(v)$. So $Z$ is a bisimulation between $(\M, w)$ and $(\M^{\div\neg\varphi}, w)$. Thus, $\M^{\div\neg\varphi}, w\vDash \varphi$. \hfill$\Box$

\bigskip
\noindent \textit{Proof of Proposition~6.3}. 
Let $\varphi,\psi,\chi, \lambda\in\Lang^\div$ with $\chi, \lambda$ propositional and consider a Kripke model $\M=(W,R,V)$ and a world $w$ in it. Notationally, let $\M^{\div\varphi}=(W^{\div\varphi},R^{\div\varphi},V^{\div\varphi})$,
$\llbracket\psi\rrbracket =\{w\in W: \M, w\vDash\psi\}$ and
$\llbracket \psi \rrbracket^{\div\varphi}=\{w\in W: \M^{\div\varphi}, w\vDash \psi\}$.\begin{enumerate}
\item \textit{Propositional Inclusion}: Assume that $\M,w\vDash [\div\varphi]B_a\chi$. Then $\M^{\div\varphi},w\vDash B_a\chi$. So $R_a^{\div\varphi}(w)\subseteq \llbracket \chi \rrbracket^{\div\varphi}$. By definition of contraction $R_a(w)\subseteq R_a^{\div\varphi}(w)\subseteq \llbracket \chi \rrbracket^{\div\varphi}$. As $\chi$ is propositional, $\llbracket \chi \rrbracket=\llbracket \chi \rrbracket^{\div\varphi}$. Hence, $\M,w\vDash B_a\chi$.

\item \textit{Propositional Conjunctive Inclusion}: Suppose $\M,w\vDash [\div \lambda\wedge\psi]\neg B_a\lambda$ and $\M,w\vDash[\div\lambda\wedge\psi] B_a\chi$. There are two cases: either $\M,w\vDash\neg B_a(\lambda\wedge \psi)$ or $\M,w\vDash B_a(\lambda\wedge \psi)$. Suppose the first. Then $R_a^{\div\lambda\wedge\psi}(w)=R_a(w)$. Since $\M,w\vDash[\div\lambda\wedge\psi]\neg B_a\lambda$, there exists $v\in R_a^{\div\lambda\wedge\psi}(w)$ with $\M^{\div\lambda\wedge\psi},v\vDash \neg \lambda$. As $\lambda$ is propositional, then $\M,v\vDash \neg \lambda$, and since $R_a^{\div\lambda\wedge\psi}(w)=R_a(w)$ then $\M,w\vDash \neg B_a\lambda$. So $R_a^{\div \lambda}(w)=R_a(w)$. Hence, $R_a^{\div\lambda\wedge\psi}(w)=R_a^{\div \lambda}(w)$. So $\M,w\vDash [\div\lambda\wedge\psi] B_a\chi$ implies $R_a^{\div \lambda}(w)=R_a^{\div\lambda\wedge\psi}(w)\subseteq \llbracket\chi\rrbracket^{\div\lambda\wedge\psi}=\llbracket\chi\rrbracket^{\div \lambda}$, where the last identity holds as $\chi$ is propositional. Hence, $\M,w\vDash [\div \lambda] B_a\chi$. Now suppose the second case holds. Then $\M,w\vDash B_a\lambda$, and so $R^{\div\lambda}_a(w)=R_a(w)\cup\llbracket\neg\lambda\rrbracket \subseteq R_a(w)\cup \llbracket\neg\lambda\vee\neg\psi\rrbracket =R_a^{\div(\lambda\wedge \psi)}(w)$. As by assumption $\M,w\vDash[\div\lambda\wedge\psi] B_a\chi$, then $R_a^{\div(\lambda\wedge \psi)}(w)\subseteq \llbracket \chi\rrbracket^{\div(\lambda\wedge \psi)}=\llbracket \chi\rrbracket^{\div\lambda}$. Hence, $R^{\div\lambda}_a(w)\subseteq\llbracket \chi\rrbracket^{\div\lambda}$ and so $\M,w\vDash[\div \lambda] B_a\chi$.

\item \textit{Propositional Conjunctive Overlap}. Suppose that $\M,w\vDash [\div \varphi] B_a\chi \wedge [\div \psi] B_a\chi$. Then $R_a^{\div\varphi}(w)\subseteq \llbracket\chi \rrbracket^{\div\varphi}$ and $R_a^{\div\psi}(w)\subseteq \llbracket\chi \rrbracket^{\div\psi}$.
As $\chi$ is propositional, $\llbracket\chi \rrbracket^{\div\varphi}=\llbracket\chi \rrbracket^{\div\psi}=\llbracket\chi \rrbracket^{\div\varphi\wedge\psi}$, and so $(R_a^{\div\varphi}(w)\cup R_a^{\div\psi}(w))\subseteq \llbracket\chi \rrbracket^{\div\varphi\wedge\psi}$. 
Two cases: either $\M,w\vDash B_a(\varphi\wedge\psi)$ or not. If the first, then $\M,w\vDash B_a\varphi$ and $\M,w\vDash B_a\psi$. So $R_a^{\div\varphi}(w)=R_a(w)\cup \llbracket\neg\varphi\rrbracket$ and $R_a^{\div\psi}(w)= R_a(w)\cup \llbracket\neg\psi\rrbracket$ and $R_a^{\div\varphi\wedge\psi}(w)= R_a(w)\cup \llbracket\neg\varphi\vee\neg\psi\rrbracket=R_a(w)\cup \llbracket\neg\varphi\rrbracket\cup\llbracket\neg\psi\rrbracket$. Then, $R_a^{\div\varphi\wedge\psi}(w)= R_a^{\div\varphi}(w)\cup R_a^{\div\psi}(w)$, and so $R_a^{\div\varphi\wedge\psi}(w)\subseteq \llbracket\chi \rrbracket^{\div\varphi\wedge\psi}$.  Hence, $\M,w\vDash [\div \varphi\wedge \psi] B_a\chi $. Suppose instead that $\M,w\vDash \neg B_a(\varphi\wedge\psi)$. Then, $R_a^{\div\varphi\wedge\psi}(w)= R_a(w)$, and since by initial assumption $R_a^{\div\varphi}(w)\subseteq \llbracket\chi \rrbracket^{\div\varphi\wedge\psi}$ and by definition of contraction $R_a(w)\subseteq R_a^{\div\varphi}(w)$, then $R_a(w)\subseteq\llbracket\chi \rrbracket^{\div\varphi\wedge\psi}$. Hence, $\M,w\vDash [\div \varphi\wedge \psi] B_a\chi$.\hfill$\Box$\end{enumerate}

\bigskip
\noindent \textit{Proof of Proposition~7.4}. Let $\M=(W,R,V)$ be a Kripke model. Notationally, we let $\E(\div\varphi)=(E,Q,Q^+,pre)$, $\M\otimes \E(\div\varphi)=(W^{\E(\div\varphi)},R^{\E(\div\varphi)},V^{\E(\div\varphi)})$ and $\M^{\div\varphi}=(W^{\div\varphi},R^{\div\varphi},V^{\div\varphi})$. The function $h\colon W^{\E(\div\varphi)}\rightarrow W^{\div\varphi}$, defined by $h((w,e))=w$, for all $w\in W$, is a bijection: first, notice that, by definition of $\E(\div\varphi)$, for each world $w\in W$ there is exactly one event $e\in E$ such that $\M,w\vDash pre(e)$, and so $(w,e)\in W^{\E(\div\varphi)}$ iff $w\in W$. Then notice that, by definition of $\div\varphi$-update, $W^{\div\varphi}=W$, and so $(w,e)\in W^{\E(\div\varphi)}$ iff $w\in W^{\div\varphi}$. We now show that $h$ preserves atomic valuations and accessibility relations. The atomic valuations are clearly preserved, as $V^{\E(\div\varphi)}(p)=V(p)=V^{\div\varphi}(p)$ for all $p\in At$. For the accessibility relations, let $a\in Ag$ and $(w,e)\in W^{\E(\div\varphi)}$. We want to show that $(v,f)\in R_a^{\E(\div\varphi)}(w,e)$ iff $v\in R^{\div\varphi}_a(w)$. We consider two cases: either $\M,w\not\vDash B_a\varphi$, or $\M,w\vDash B_a\varphi$. Suppose $\M,w\not\vDash B_a\varphi$. Then $\E(\div\varphi)$ is such that $Q_a(e)=E$ and $Q^+_a(e)=\emptyset$, and $\M^{\div\varphi}$ is such that $R_a(w)=R_a^{\div\varphi}(w)$. So we have that $(v,f)\in R^{\E(\div\varphi)}_a(w,e)$ iff (by definition of generalized product update) $v\in R_a(w)$ iff (by definition of $\div\varphi$-update) $v\in R^{\div\varphi}_a(w)$, as we wanted to show. Suppose instead $\M,w\vDash B_a\varphi$. Assume $(v,f)\in R^{\E(\div\varphi)}_a(w,e)$. Then either $v\in R_a(w)$ and $f\in Q_a(e)$, or $f\in Q^+_a(e)$. If $v\in R_a(w)$, then $v\in R^{\div\varphi}_a(w)$, as desired. If instead $f\in Q^+_a(e)$, then $\M,v\vDash \neg \varphi$ and so $v\in R^{\div\varphi}_a(w)$, as desired. For the other direction, assume $v\in R^{\div\varphi}_a(w)$. Then either $v\in R_a(w)$ or $\M,v\vDash \neg\varphi$. If $v\in R_a(w)$, then since $Q_a(e)=E$, $(v,f)\in R^{\E(\div\varphi)}_a(w,e)$. If instead $\M,v\vDash \neg\varphi$, then by $(v,f)\in W^{\E(\div\varphi)}$, $pre(f)\vDash \neg\varphi$ by definition of $\E(\div\varphi)$. Since $\M,w\vDash B_a\varphi$ then $Q^+_a(e)=\{g\in E\colon pre(g)\vDash \neg\varphi\}$, and so $(v,f)\in R^{\E(\div\varphi)}_a(w,e)$. Hence in all cases we have the desired.\hfill$\Box$


\bigskip
\noindent \textit{Proof of Theorem~7.7}. It is straightforward that the inference rules preserve validity. We check the soundness of A4 and A5. The soundness of the other axioms is immediate. Let $\M=(W,R,V)$ be a Kripke model, let $\E=(E,Q,Q^+,pre)$ be a generalized event model, and let $\M\otimes\E=(W^\E,R^\E,V^\E)$. 
	    \begin{itemize}
	        \item A4: Note that, if $\M, w\not\vDash pre(e)$, then the biconditional holds automatically at $(\M, w)$. Suppose $\M, w\vDash pre(e)$ and $\M, w\vDash [\E, e]\forall\varphi$. Then $\M\otimes\E, (w,e)\vDash\forall\varphi$, i.e. for all $(v,f)\in W^\E$, $\M\otimes\E, (v,f)\vDash\varphi$. In other words, for all $f\in E$ and $v\in W$ such that $\M, v\vDash pre(f)$, $\M, v\vDash[\E,f]\varphi$. So $\M,w\vDash\bigwedge_{f\in E}\forall[\E,f]\varphi$.
	
	        Conversely, suppose $\M, w\vDash pre(e)$ but $\M, w\vDash \neg [\E,e]\forall\varphi$, i.e. $\M\otimes\E, (w,e)\not\vDash\forall\varphi$. So there is $(v,f)\in W^\E$ such that $\M\otimes\E, (v,f)\not\vDash\varphi$. It follows that $\M, v\vDash \neg [\E,f]\varphi$ and so $\M, w\vDash\neg\bigwedge_{f\in E}\forall [\E, f]\varphi$.
	        \item A5: 
           Suppose $\M, w\vDash [\E, e]B_a\varphi\wedge pre(e)$. Let $f\in Q_a^+(e)$. Then for any $v\in W$, if $\M,v\vDash pre(f)$, we have $(v,f)\in R_a^\E((w,e))$ and so by assumption $\M^\E, (v,f)\vDash \varphi$. Thus $\M, w\vDash \bigwedge_{f\in Q_a^+(e)}\forall [\E,f]\varphi$. Let $f\in Q_a(e)$ and $v\in R_a(w)$. Then $(v,f)\in R_a^\E((w,e))$ and so $\M^\E, (v,f)\vDash \varphi$. Thus $\M, v\vDash[\E,f]\varphi$. So $\M, w\vDash \bigwedge_{f\in Q_a(e)}B_a[\E,f]\varphi$.

Conversely, suppose $\M, w\vDash pre(e)\rightarrow\bigwedge_{f\in Q_a^+(e)}\forall [\E,f]\varphi\wedge\bigwedge_{f\in Q_a(e)}B_a[\E,f]\varphi$ and let $(v,f)\in R_a^\E((w,e))$. Either $f\in Q_a^+(e)$ or $f\in Q_a(e)$ and $v\in R_a(w)$. In the first case, since $\M, w\vDash \forall [\E,f]\varphi$, we have $\M, v\vDash[\E,f] \varphi$. In the second case, since $\M, w\vDash B_a[\E,f]\varphi$, we have $\M, v\vDash [\E,f]\varphi$. So either way, $\M^\E, (v,f)\vDash\varphi$ and so $\M^\E, (w,e)\vDash B_a\varphi$. 

\item A6: The proof mainly requires showing that there exists an isomorphism between $(\M\otimes\E)\otimes\F=(W^{\E\F},R^{\E\F},V^{\E\F})$ and  $\M\otimes(\E\circ\F)=(W^{\E\circ\F},R^{\E\circ\F},V^{\E\circ\F})$. By invariance of truth under isomorphism, this will imply the desired. We use notation $\M\otimes\E=(W^{\E},R^{\E},V^{\E})$, and $\E=(E^\E,Q^\E,Q^{+\E},pre^\E)$ and $\F=(E^\F,Q^\F,Q^{+\F},pre^\F)$.
			
To show the existence of the isomorphism, let $h: W^{\E\F}\rightarrow W^{\E\circ\F}$ be defined by $h((w,e),f)=(w,(e,f))$. Clearly, $h$ is a bijection: $((w,e),f)\in W^{\E\F}$
iff $\M,w\models pre^\E(e)
\quad\text{and}\quad
\M\otimes\E,(w,e)\models pre^\F(f),$ iff
$\M,w\models pre^\E(e)\wedge[\E,e]pre^\F(f),
$ iff $(w,(e,f))\in W^{\E\circ\F}$.
It also preserves valuations, since for every atom $p\in At$, $((w,e),f)\in V^{\E\F}(p)$ iff $w\in V(p)$ iff $(w,(e,f))\in V^{\E\circ\F}(p)$. Finally, $h$ preserves accessibility relations. 
Let $((w',e'),f')\in R^{\E\F}_a(((w,e),f))$. We want to show that $(w',(e',f'))\in R^{\E\circ\F}_a((w,(e,f)))$. Let $((w',e'),f')\in R^{\E\F}_a(((w,e),f))$. Two cases: either $f'\in Q_a^{+\F}(f)$ or $(w',e')\in R^\E_a((w,e))$ and $f'\in Q_a^{\mathcal{F}}(f)$. If $f'\in Q_a^{+\F}(f)$ then $(e', f')\in Q^{+\E\circ\F}_a(e,f)$ and so $(w',(e',f'))\in R^{\E\circ\F}_a((w,(e,f)))$. If $(w',e')\in R^\E_a((w,e))$ and $f'\in Q_a^{\mathcal{F}}(f)$, then again two cases, either $e'\in Q_a^{+\E}(e)$, or $w'\in R_a(w)$ and $e'\in Q_a^{\E}(e)$. If $e'\in Q_a^{+\E}(e)$ then by $f'\in Q_a^{\F}(f)$ and definition of composition, $(e', f')\in Q^{+\E\circ\F}_a((e,f))$, which implies $(w',(e',f'))\in R^{\E\circ\F}_a((w,(e,f)))$. If $w'\in R_a(w)$ and $e'\in Q_a^{\mathcal{E}}(e)$, then by $f'\in Q_a^{\mathcal{F}}(f)$ and definition of composition we have $(e', f')\in Q^{\E\circ\F}_a((e,f))$, and so $(w',(e',f'))\in R^{\E\circ\F}_a((w,(e,f)))$. Hence, in all cases the desired holds. The other direction (that is, if $(w',(e',f'))\in R^{\E\circ\F}_a((w,(e,f)))$ then $((w',e'),f')\in R^{\E\F}_a((w,e),f)$) follows by a similar straightforward application of the definition of composition and generalized product update.

Thus, $h$ is an isomorphism. By invariance of truth under isomorphisms, for every $((w,e),f)\in W^{\E\F}$ and every $\varphi\in \Lang_{\text{DEL}}$, we have $((\M\otimes\E)\otimes\F,((w,e),f))\vDash\varphi$ iff $(\M\otimes(\E\circ\F),(w,(e,f)))\vDash\varphi$. It remains only to connect this with the truth conditions for the dynamic modalities. Let $w\in W$. If $\M,w\not\vDash pre^\E(e)$, then both $[\E,e][\F,f]\varphi$ and $[\E\circ\F,(e,f)]\varphi$ are true at $w$ vacuously, since $pre^{\E\circ\F}(e,f)=pre^\E(e)\wedge[\E,e]pre^\F(f)$. If $\M,w\vDash pre^\E(e)$ but $\M\otimes\E,(w,e)\not\vDash pre^\F(f)$, then again $[\E,e][\F,f]\varphi$ is true at $w$ vacuously, and so is $[\E\circ\F,(e,f)]\varphi$, since $\M,w\not\vDash pre^{\E\circ\F}(e,f)$. Finally, if $\M,w\vDash pre^\E(e)$ and $\M\otimes\E,(w,e)\vDash pre^\F(f)$, then both product points $((w,e),f)$ and $(w,(e,f))$ exist, and the desired equivalence follows from the isomorphism above. Hence, in all cases, $\M,w\vDash[\E,e][\F,f]\varphi$ iff $\M,w\vDash[\E\circ\F,(e,f)]\varphi$. Since $\M$ and $w$ were arbitrary, the equivalence is valid.
	    \end{itemize}
	    Completeness of \textsf{GDEL} follows from the completeness of \textsf{K} with the universal modality by standard reduction arguments.\hfill$\Box$

\bigskip
\noindent \textit{Proof of Theorem~7.8}. 
		Let $\M=(W,R,V)$ be a Kripke model and $\E=(E,Q,Q^+,pre)$ be a generalized event model. Notationally, let $\M\otimes \E=(W^\E,R^\E,V^\E)$. Our proof strategy is to define a submodel of $\M\otimes\E$ that serves as the intermediate  refinement step, and then show that $\M\otimes\E$ is a simulation of that model.
		Let $(W^{\E^-}, R^{\E^-},V^{\E^-})$ be a Kripke model such that $W^{\E^-}= W^\E$, $R_a^{\E^-}(w,e)=\{(v,f)\in W^\E\colon v\in R_a(w), f\in Q_a(e)\}$, for all $a\in Ag$, and $V^{\E^-}=V^\E$. This is a submodel of $\M\otimes \E$, only missing the edges between worlds that the $Q^+$ relation added. It is then clear that we recover the full $R^\E$ by taking the union of $R^{\E^-}$ with the set containing those edges, that is
				$R_a^{\E}(w,e)=R_a^{\E^-}(w,e)\cup \{(v,f)\in W^\E\colon f\in Q^+_a(e)\}$, for all $a\in Ag$ and $(w,e)\in W^\E$. 
		Now, $(W^{\E^-}, R^{\E^-},V^\E)$ is a standard \textsf{DEL} update, and it is a well known result that it is a refinement of $\M$ \cite{bozzelli2014refinement}. We now show that $\M\otimes\E$ is a simulation of $\M\otimes\E^-$, from which we can conclude that $\M\otimes\E$ is a simulation of a refinement of $\M$. Let 
		$Z\subseteq W^{\E^-}\times W^\E$ be such that $((w,e),(w,e))\in Z$. As $R^\E$ only adds edges to $R^{\E^-}$ without removing any, it clearly holds that if $(v,f)\in R^{\E^-}_a(w,e)$, then $(v,f)\in R^{\E}_a(w,e)$, and since $((v,f),(v,f))\in Z$ then $Z$ is a simulation. \hfill$\Box$
	\bibliography{bib}
    	\bibliographystyle{eptcs}
\end{document}